\documentclass[letterpaper, 12pt, draftcls, onecolumn]{ieeeconf}

\IEEEoverridecommandlockouts                              

\usepackage[dvipdfmx]{graphicx}
\usepackage{amsmath, amssymb, bm, mathrsfs}
\usepackage{cite, color, graphicx, subcaption}

\newtheorem{theorem}{Theorem}[section]
\newtheorem{proposition}[theorem]{Proposition}

\newtheorem{problem}[theorem]{Problem}
\newtheorem{assumption}[theorem]{Assumption}
\newtheorem{example}[theorem]{Example}
\newtheorem{lemma}[theorem]{Lemma}
\newtheorem{remark}[theorem]{Remark}
\newtheorem{corollary}[theorem]{Corollary}
\newtheorem{algorithm}[theorem]{Algorithm}

\title{\LARGE \bf
Stabilization of Systems with Asynchronous 
Sensors and Controllers\thanks{}
}

\author{Masashi Wakaiki, Kunihisa Okano, and Jo\~ao~P.~Hespanha
\thanks{
This paper was partially
presented at the 2015 American Control Conference,
July 1-3, 2015, the USA.
This material is based upon work supported by 
the National Science
Foundation under Grant No. CNS-1329650.
M. Wakaiki acknowledges Murata Overseas Scholarship Foundation
for the support of this work.
K. Okano is supported by 
JSPS Postdoctoral Fellowships for Research Abroad.
Corresponding author M. Wakaiki. Tel. +1 805 893 7785. 
Fax +1 805 893 3262.
}
\thanks{
M. Wakaiki is with the Department of Electrical and Electronic Engineering, Chiba University,
Chiba, 263-8522, Japan (e-mail:{\tt \small wakaiki@chiba-u.jp}).
K. Okano is with the Graduate School of Natural Science and Technology,
Okayama University, Okayama, 700-8530, Japan 
(e-mail:{\tt \small kokano@okayama-u.ac.jp}).
J.~P.~Hespanha is with the Center for Control,
Dynamical-systems and Computation (CCDC), University of California,
Santa Barbara, CA 93106-9560 USA
(e-mail:{\tt \small hespanha@ece.ucsb.edu}).
}%
}

\begin{document}

\maketitle
\thispagestyle{empty}
\pagestyle{empty}

\begin{abstract}
We study the stabilization of networked control systems with
asynchronous sensors and controllers. 
Offsets between the sensor and controller clocks
are unknown and
modeled as parametric uncertainty.
First we consider multi-input linear systems and
provide a sufficient condition 
for the existence of linear time-invariant controllers that are 
capable of stabilizing the closed-loop system for 
every clock offset in a given range of admissible values. 
For first-order systems, we next
obtain the maximum length of the offset range
for which the system can be stabilized 
by a single controller. 
Finally,
this bound is compared with the offset bounds that would be allowed if
we restricted our attention to static output feedback controllers.
\end{abstract}

\section{Introduction}
In networked and embedded control systems,
the outputs of plants are often sampled in a nonperiodic fashion and sent to
controllers with time-varying delays.
To address robust control with such imperfections,
various techniques have been developed, for example, the input-delay approach \cite{Fridman2004, Mirkin2007},
the gridding approach \cite{Fujioka2009, Oishi2010, Donkers2011},
and the impulsive systems approach based on
Lyapunov functionals \cite{Naghshtabrizi2010}, on
looped functionals \cite{Briat2012}, and on 
clock-dependent Lyapunov functions \cite{Briat2013};
see also the surveys \cite{Hespanha2007, Hetel2017}.
In contrast to the references mentioned above, 
here we assume that time-stamps are used to provide the
controller with information about the sampling times and the communication delays
incurred by each measurement.
In this approach,
sensors send measurements to controllers
together with time-stamps,
and the controllers
exploit this information to mitigate the effect of variable delays and sampling periods
\cite{Graham2004, Nakamura2008, Garcia2014}.
However, when the local clocks at
the sensors and at the controllers are not synchronized, 
the time-stamps and 
the true sampling instants 
do not match.
Protocols to establish synchronization have been actively
studied as surveyed in \cite{Rhee2009}, 
and synchronization by the global positioning system
(GPS) or radio clocks has been utilized in some systems.
Nevertheless, synchronizing clocks over networks 
has fundamental limits~\cite{Freris2011}, and
a recent study \cite{Jiang2013} 
has shown that synchronization
based on GPS signals is vulnerable against attacks.

In this paper, we study the stabilization problem of
systems with asynchronous sensing and control.
We assume that
the controller can use the time-stamps but does not know
the offset between
the sensor and controller clocks, but
we do assume that this offset is essentially constant over the time
scales of interest.
Our objective is to find linear time-invariant (LTI) controllers 
that achieve closed-loop stability
for every clock offset in a given range.


We formulate the stabilization of systems with clock offsets as the problem
of stabilizing systems with parametric uncertainty,
which can be regarded as the
simultaneous stabilization of a family of plants, as 
studied in \cite[Sec. 5.4]{vidyasagar1985} and
\cite{Vidyasagar1982}.
However, we had to overcome a few technical difficulties that distinguish the problem considered here from previously published results:

{\it Infinitely many plants:}
We consider a family of plant models that is indexed by 
a continuous-valued parameter.
Such a family
includes infinitely many plants, but
the approaches for simultaneous stabilization 
e.g., in \cite{Shi2009}
exploit the property that the number of plant models is finite.

{\it Nonlinearity of the uncertain parameter:}
In this work, 
the uncertain parameter appears in a non-linear form.
Therefore, it is not suitable to use the techniques 
based on linear matrix inequalities (LMIs) in 
\cite{Oliveira1999}
for the robust stabilization
of systems with polytopic uncertainties.
Although the robust stability analysis based on continuous paths of
systems with respect to the $\nu$-gap metric
was developed in \cite{Cantoni2012},
controller designs based on this approach have not been fully investigated.

{\it Common unstable poles and zeros:}
Earlier studies on simultaneous stabilization consider
a restricted class of plants.
For example, 
the sufficient condition 
in \cite{Blondel1993Suf}
is obtained for a family of plants with
no common unstable zeros or poles. 
The set of plants in \cite{Maeda1984} has
common unstable zeros (or poles) but
all the plants are stable (or minimum-phase).
These assumptions are not satisfied for the systems 
in the present paper.

We make the following technical contributions for 
multi-input systems and first-order systems:
First we consider multi-input systems and 
obtain a sufficient condition for stabilization with asynchronous sensing and control.
We construct a stabilizing controller from
the solution of an appropriately defined  $\mathcal{H}^{\infty}$ control problem.
The above mentioned difficulties found in 
the simultaneous stabilization problem we consider 
is circumvented by exploiting geometric properties on $\mathcal{H}^{\infty}$.
For first-order systems,
we obtain
an explicit formula for the exact bound 
on the clock offset that can be allowed for stability.
This result is based on the stabilization of
interval systems \cite{Ghosh1988, Olbrot1994}, 
to which our problem can be reduced for first-order plants.
We start by formulating the problem in the context of state feedback without
disturbances and noise, but we show in Section 3.2 that
the above results also apply for output feedback with disturbances and 
noise.


The authors in the previous study \cite{Okano2015}
have considered systems with time-varying clock offsets
and have
proposed a stabilization method with causal controllers, 
based on the analysis of data rate limitations in quantized control.
The stability analysis and the $\mathcal{L}^2$-gain analysis of systems with variable 
clock offsets have been investigated in \cite{WakaikiCDC2015} and
\cite{Wakaiki2016ACC}, respectively. 
The major difference with respect to those studies is that
here we consider only constant offsets but
design stabilizing LTI controllers. 
This paper is based on the conference paper \cite{Wakaiki2015ACC}, but
here we extend the preliminary results for single-input systems to
the multi-input case.

The remainder of the paper is organized as follows.
Section 2 introduces the closed-loop system we consider 
and presents the problem formulation.
Section 3 is devoted to the discretization of the closed-loop system.
In Section~4, we obtain
a sufficient condition for the stabilizability
of general-order systems.
In Section 5,
we derive the exact bound
on the permissible clock offset for first-order systems. 
In Section 6, we discuss
stabilizability with static controllers and
the comparison of 
the offset bounds obtained for LTI controllers 
and static controllers.

{\it Notation and definitions:~}
We denote by $\mathbb{Z}_+$ the set of non-negative integers.
The symbols 
$\mathbb{D}$, $\bar{\mathbb{D}}$, and $\mathbb{T}$ denote
the open unit disc $\{ z \in \mathbb{C}:~ |z| < 1\}$,
the closed unit disc $\{ z \in \mathbb{C}:~ |z| \leq 1\}$, and
the unit circle $\{ z \in \mathbb{C}:~ |z| = 1\}$, respectively.
We denote by $\mathbb{D}^c$ the complement of the open unit disc
$\{ z \in \mathbb{C}:~ |z| \geq  1\}$.

A square matrix $F$ is said to be {\em Schur stable}
if all its eigenvalues lie in the unit disc $\mathbb{D}$.
We say that a discrete-time LTI system $\xi_{k+1} = F\xi_k + Gu_k,~y_k = H \xi_k$
is stabilizable (detectable) if 
there exists a matrix $K$ ($L$) such that $F-GK$ ($F-LH$) is Schur stable.
We also use the terminology $(F, G)$ is 
stabilizable (respectively, $(F, H)$ is detectable) to denote
this same concept.

We denote by $\mathcal{RH}^{\infty}$ the space of all bounded holomorphic
real-rational
functions in $\mathbb{D}$.
The field of fractions of 
$\mathcal{RH}^{\infty}$ is denoted by $\mathcal{RF}^{\infty}$.
For a commutative ring $R$, 
$\mathbf{M}(R)$ denotes the set of
matrices with entries in $R$, of any order.
For $M \in \mathbf{M}(\mathbb{C})$,
$\|M\|$ denotes the induced 2-norm.
For $G \in \mathbf{M}(\mathcal{RH}^{\infty})$,
the $\mathcal{RH}^{\infty}$-norm is defined as $\|G\|_{\infty} = 
\sup_{z\in\mathbb{D}} \|G(z)\|$.
For $G = \begin{bmatrix}
G_{11} &\quad  G_{12} \\
G_{21} &\quad  G_{22}
\end{bmatrix} \in \mathbf{M}(\mathcal{RF}^{\infty})$ and
$Q\in \mathbf{M}(\mathcal{RF}^{\infty})$,
we define a lower linear fractional transformation of $G$ and $Q$
as
$\mathcal{F}_{\ell}(G, Q) := G_{11} + G_{12}Q(I-G_{22}Q)^{-1}G_{21}$.

A pair $(N, D)$ in $\mathbf{M}(\mathcal{RH}^{\infty})$ is said to be 
{\em right coprime} if  the Bezout identity 
$	X N + Y D=I $
holds for some $X$, $Y \in \mathbf{M}(\mathcal{RH}^{\infty})$.
$P \in \mathbf{M}(\mathcal{RF}^{\infty})$ admits a {\em right coprime
	factorization} if there exist $D$, $N \in \mathbf{M}(\mathcal{RH}^{\infty})$ 
such that 
$P = ND^{-1}$ and the pair $(N,D)$ is right coprime.
Similarly, a pair $(\tilde D,\tilde N)$ 
in $\mathbf{M}(\mathcal{RH}^{\infty})$ 
is {\em left coprime} if the Bezout identity 
$	\tilde N \tilde X + \tilde D \tilde Y=I$
holds for some $\tilde X$, $\tilde Y \in \mathbf{M}(\mathcal{RH}^{\infty})$.
$P \in \mathbf{M}(\mathcal{RF}^{\infty})$ admits a {\em left coprime
	factorization} if there exist $\tilde D$, $\tilde N \in
\mathbf{M}(\mathcal{RH}^{\infty})$ 
such that $P = \tilde D^{-1} \tilde N$ and the pair 
$(\tilde D,\tilde N)$ is left coprime.
If $P$ is a scalar-valued function, then we use the expressions {\em coprime} and
{\em coprime factorization}.

\section{Problem Statement}
Consider the following LTI plant:
\begin{equation}
\label{eq:original_plant}
\Sigma_P:~
\dot x(t) = Ax(t) + Bu(t),
\end{equation}
where $x(t) \in \mathbb{R}^n$ and
$u(t) \in \mathbb{R}^m$ are the state and the input of the plant, respectively.
As shown in 
Fig.~\ref{fig:CLS},
this plant is connected 
through a sampler and a zero-order hold (ZOH)
to a time-stamp aware estimator and a controller,
which will be described soon.
\begin{figure}[tb]
	\centering
	\includegraphics[width = 8cm]{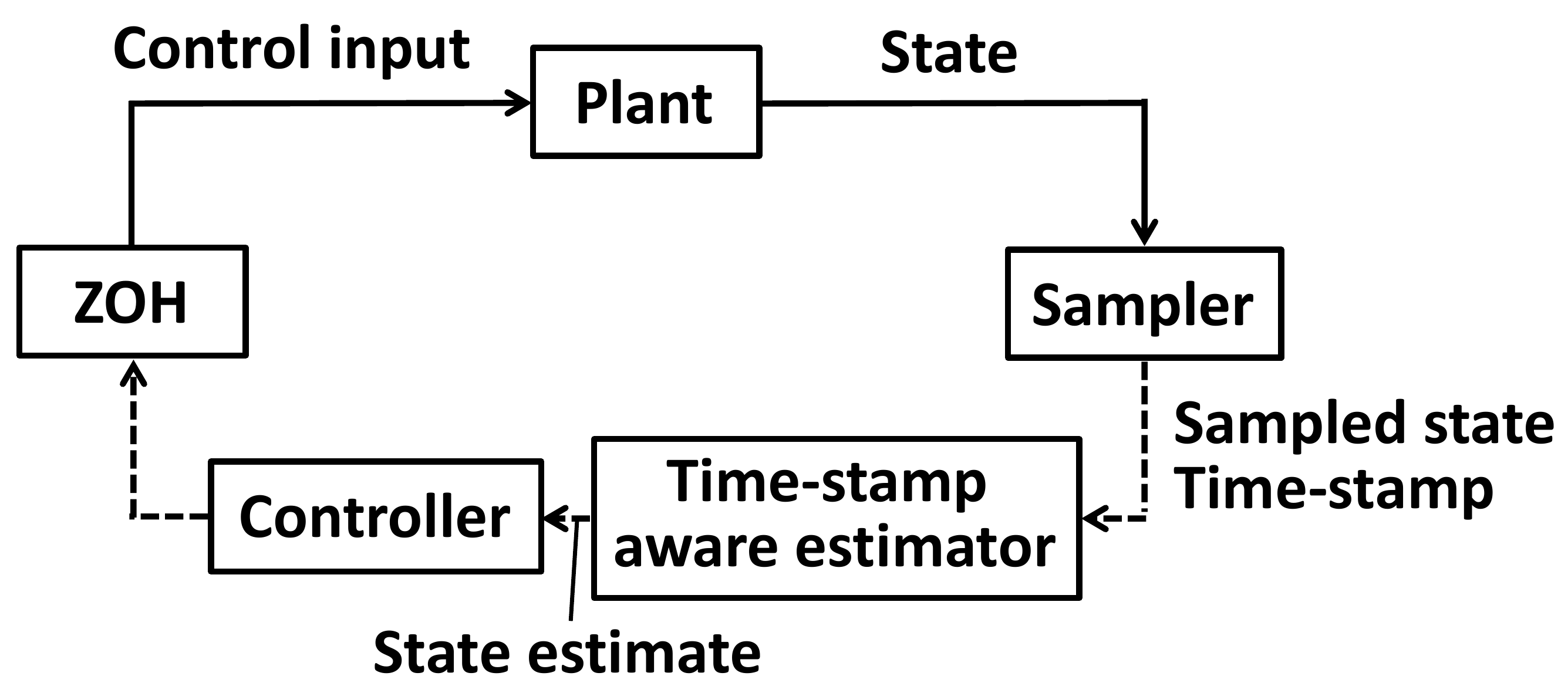}
	\caption{Closed-loop system with a time-stamp aware estimator.}
	\label{fig:CLS}
\end{figure}

Let $s_1,s_2,\dots$ be sampling instants
from the perspective of the controller clock.
A sensor measures the state $x(s_k)$ and sends it to a controller
together with a time-stamp.
However, since the sensor and the controller may not be synchronized,
the time-stamp determined by the sensor 
typically includes an 
unknown offset with respect to the controller clock. 
In this paper, we assume that the clock offset is constant. 
Although
clock properties are affected by environment such as temperature
and humidity,
the change of such properties is slow 
for the time scales of interest.
Furthermore, the difference of 
clock frequencies can be ignored. 
This is justified by noting that 
time synchronization techniques, like the one proposed 
in \cite{He2014},
can achieve asymptotic convergence of the clock frequencies (in the mean-square sense),
even in the presence of random network delays.
We thus assume that
the time-stamp $\hat s_k$ reported by the sensor is given by
\begin{equation}
\label{eq:offset_eq}
\hat s_k = s_k + \Delta \qquad (k \in \mathbb{N})
\end{equation}
for some unknown constant $\Delta \in \mathbb{R}$.

Let $h > 0$ be the update period of the ZOH.
The control signal $u(t)$ 
is assumed to be piecewise constant and updated periodically at times $t_k=k h$ ($k\in \mathbb{N}$) with values $u_k$ computed by the controller:
$u(t) = u_k$ for $t \in [t_k,t_{k+1})$.
We place a basic assumption for stabilization of sampled-data systems.
\begin{assumption}
	\label{ass:stabilizability}
	{\em(}Stabilizability and non-pathological control update{\em)}
	{\em
		The plant $(A,B)$ is stabilizable and the update period $h$ is non-pathological,
		that is, $(\lambda_p - \lambda_q) h \not= 2\pi j \ell$ ($\ell = \pm 1, \pm 2,\dots$)
		for each pair $(\lambda_p,\lambda_q)$ of eigenvalues of $A$.
	}
\end{assumption}

While the ZOH updates the control signal $u(t)$ periodically, 
the true sampling times $s_k$ and the reported sampling times 
$\hat s_k$ may not be periodic. 
However, we do assume that both $s_k$ and $\hat s_k$ do not fall 
behind $t_k$ by more than the ZOH update period $h$.
This assumption is formally stated as follows.
\begin{assumption}
	\label{ass:sample}
	{\em(}Bounded clock offset{\em)}
	{\em
		For every $k \in \mathbb{Z}_+$,
		$s_k, \hat s_k \in[ t_{k}, t_{k+1})$. 
	}
\end{assumption}
This assumption implies that
the clock offset $\Delta$ is smaller than the control update period $h$,
which holds 
in most mechatronics systems. In fact,
control update periods 
for mechatronics systems
generally take values from 100 $\mu$s to 10 ms,
while
recent clock
synchronization algorithms such as the IEEE 1588 Precision Time Protocol (PTP)
\cite{PTP}
make clock offsets smaller than a few tens of microseconds.

\begin{figure}[tb]
	\centering
	\includegraphics[width = 8.7cm,clip]{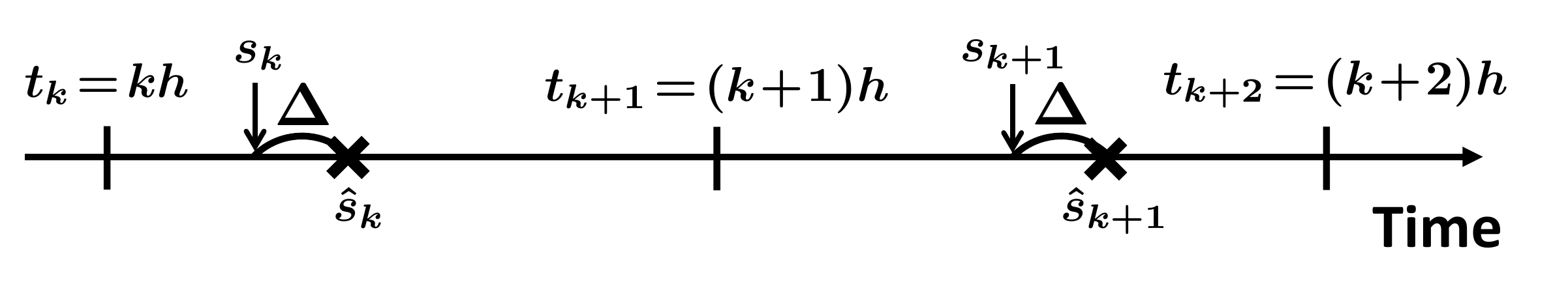}
	\caption{Sampling instants $s_k$, reported time-stamps $\hat s_k$, 
		and updating instants $t_k$ of the zero-order hold.}
	\label{fig:clockmismatch}
\end{figure}

Fig.~\ref{fig:clockmismatch} shows the timing diagram
of the sampling instants $s_k$,
the reported time-stamps $\hat s_k$, 
and updating instants $t_k$ of the control inputs.

The controller side is comprised of a time-stamp aware estimator and a
controller
as in the model-based or emulation-based control of 
networked control systems \cite{Garcia2014}.
The time-stamp aware estimator generates 
the state estimate $\hat{x}(t_{k+1}) \in \mathbb{R}^n$
from the data $(x(s_k), \hat s_k)$ according to the following dynamics:
\begin{equation}
\label{eq:estimator}
\Sigma_E:~
\begin{cases}
\dot{\hat x}(t) = A\hat{x}(t) + Bu(t)
&\qquad(t_k < t \leq t_{k+1})	 \\
\hat{x}(\hat s_k) = x(s_k),
&\qquad(k \in \mathbb{Z}_+).
\end{cases}
\end{equation} 
Note that if the time-stamp is correct, i.e., 
$s_k=\hat s_k$, then this estimator consistently produces 
$\hat x(t) = x(t)$ for all $t$, perfectly compensating
transmission delays.
Time-stamp aware estimators have been used to compensate for
network-induced imperfections, e.g., in \cite{Graham2004,Nakamura2008,Garcia2014}.

The controller is a discrete-time LTI system and 
generates the control input $u_k$ based on the state estimate $\hat x_k := \hat x(t_k)$:
\begin{equation}
\label{eq:controller}
\Sigma_C:~
\begin{cases}
\zeta_{k+1} = A_c\zeta_{k} + B_c\hat x_k \\
u_k = C_c \zeta_{k} + D_c \hat x_k,
\end{cases}
\end{equation} 
where $\zeta_{k} \in \mathbb{R}^{n_c}$ is the state of the controller.

The objective of the present paper is to find
a discrete-time LTI controller $\Sigma_C$ as in \eqref{eq:controller} 
that achieves closed-loop stability for every
clock offset in a given range of admissible values.
Specifically, we want to solve the following problem:
\begin{problem}
	\label{problem:clockoffset}
	{\it
		Given an offset interval $[\underline \Delta, \overline \Delta]$,
		determine if there exists a controller $\Sigma_C$ as in \eqref{eq:controller} such that
		$x(t)$, $\hat x(t) \to 0$ as $t \to \infty$ and
		$\zeta_{k} \to 0$ as $k \to \infty$ 
		for every $\Delta \in [\underline \Delta, \overline \Delta]$
		and for every initial states $x(0)$ and $\zeta_0$.
		Furthermore, if one exists, find such a controller $\Sigma_C$.
	}
\end{problem}

\section{Discretization of the Closed-loop System}
To solve Problem \ref{problem:clockoffset},
we discretize the system comprised of 
the plant $\Sigma_P$, the estimator $\Sigma_E$, the ZOH, and the sampler.
In this section, we obtain a realization for the discretized system and 
describe its basic properties related to stability, stabilizability, and detectability.
Moreover, we extend the discretized system to scenarios with
disturbances/noise and output feedback.
\subsection{Discretized system and its basic properties}
The following lemma provides a realization for the discretized system:
\begin{lemma}
	\label{lem:discretization}
	Define 
	\begin{equation*}
	\xi_k :=  
	\begin{bmatrix}
	x(t_k) - \hat x(t_k) \\
	\hat x(t_k)
	\end{bmatrix}.
	\end{equation*}
	The dynamics of 
	the discretized system $\Sigma_d$ 
	comprised of the plant $\Sigma_P$, the estimator $\Sigma_E$, 
	the ZOH, and the sampler
	can be described by the following equations:
	\begin{equation}
	\label{eq:extended_sys}
	\Sigma_d:~
	\xi_{k+1} = F_{\Delta} \xi_k + G_{\Delta} u_k,\qquad \eta_k = H_{\Delta}\xi_k,
	\end{equation}
	where $\Lambda := e^{Ah}$,
	$\Theta := e^{-A\Delta} - I$, and
	\begin{align}
	F_{\Delta} &\hspace{-1pt}:=\hspace{-1pt}
	\begin{bmatrix}
	-\Lambda \Theta &\quad  -\Lambda\Theta \\
	\Lambda (I + \Theta)  &\quad  \Lambda (I + \Theta) 
	\end{bmatrix},\quad
	H_{\Delta}  \hspace{-1pt}:=\hspace{-1pt}
	\begin{bmatrix}
	0 &\quad I
	\end{bmatrix}
	\notag \\
	G_{\Delta} &\hspace{-1pt}:=\hspace{-1pt}
	\begin{bmatrix}
	\Lambda \hspace{-1pt}
	\left(
	\int^h_0 e^{-A\tau} d\tau
	\hspace{-1pt}-\hspace{-1pt}
	(I \hspace{-1pt}+\hspace{-1pt} \Theta) \hspace{-1pt}
	\int^{h-\Delta}_0 e^{-A\tau} d\tau
	\right) \hspace{-1pt}
	B  \\
	\Lambda (I + \Theta)   \int^{h-\Delta}_0 e^{-A\tau} d\tau B 
	\end{bmatrix}. \label{eq:FGH_def} 
	\end{align}
\end{lemma}
\begin{proof}
	Using $\Lambda = e^{Ah}$,
	we have from the state equation
	\eqref{eq:original_plant} that
	\begin{align}
	\label{eq:discretiztion_plant}
	x(t_{k+1}) 
	&=
	\Lambda x(t_k) + \Lambda 
	\int^{h}_{0} e^{-A\tau}  d\tau \cdot Bu_k.
	\end{align}
	
	We compute $\hat x(t_{k+1})$ in terms of $x(t_{k})$ and $u_k$.
	It follows from 
	the dynamics of the estimator $\Sigma_E$ in \eqref{eq:estimator} that
	\begin{equation}
	\label{eq:hat_x_t}
	\hat x(t_{k+1}) = 
	e^{A(t_{k+1} - \hat s_k)} \hat x(\hat s_{k})
	+ 
	\int^{t_{k+1}}_{\hat s_k} e^{A(t_{k+1} - \tau)} B d\tau \cdot u_k
	\end{equation}
	and
	\begin{equation}
	\label{eq:hat_x_hat_s}
	\hat x(\hat s_{k}) = x(s_k) = 
	e^{A(s_k-t_k)} x(t_k) + 
	\int^{s_{k}}_{t_k} e^{A(s_k - \tau)} B d\tau \cdot u_k.	
	\end{equation}
	Since $t_{k+1} - t_{k} = h$ and $\hat s_{k} = s_k + \Delta$, it follows that
	\begin{equation}
	\label{eq:exp_identity}
	e^{A(t_{k+1} - \hat s_k)} \cdot e^{A(s_k-t_k)}
	= e^{A(h-\Delta)},
	\end{equation}
	and also that 
	\begin{align}
	\label{eq:int_indentity}
	e^{A(t_{k+1} - \hat s_k)}
	\int^{s_{k}}_{t_k} e^{A(s_k - \tau)} d\tau
	=
	\int^{\hat s_{k}}_{t_k+\Delta} e^{A(t_{k+1}  - \tau)} d\tau.
	\end{align}
	Using $\Lambda = e^{Ah}$ and $\Theta = e^{-A\Delta} - I$,
	we conclude 
	from \eqref{eq:hat_x_t}--\eqref{eq:int_indentity} that
	\begin{align}
	\hat x(t_{k+1}) &= 
	\Lambda (I + \Theta) x(t_k) + 
	\Lambda (I + \Theta) \int^{h-\Delta}_{0} e^{-A\tau} B d\tau \cdot u_k.
	\label{eq:hat_x_t_final}
	\end{align}
	
	From \eqref{eq:discretiztion_plant} and \eqref{eq:hat_x_t_final},
	we obtain the $F_{\Delta}$ and $G_{\Delta}$ in \eqref{eq:FGH_def}.
	Moreover, 
	we have $H_{\Delta} = [0 ~~I]$
	by the definition of the extended state $\xi_k$. 
	\hspace*{\fill} $\Box$
\end{proof}

Next we show that if the extended state $\xi_k$ and the controller state $\zeta_{k}$
converge to the origin, then
the intersample values of $x$ and $\hat x$ also converge to the origin.
\begin{proposition}
	\label{prop:stability_equivalence}
	For the discreteized system $\Sigma_d$ in Lemma \ref{lem:discretization}, we have
	that $\xi_k, \zeta_{k} \to 0$ as $k \to \infty$ if and only if
	$x(t),\hat x(t) \to 0$ as $t \to \infty$ and $\zeta_{k} \to 0$ as $k \to \infty$.
\end{proposition}
\begin{proof}
	The statement that 
	$x(t),\hat x(t) \to 0$ as $t \to \infty$ and $\zeta_{k} \to 0$ as $k \to \infty$
	imply 
	$\xi_k, \zeta_{k} \to 0$ as $k \to \infty$, follows directly
	from the definition of $\xi_k$.
	
	To prove the converse statement,
	assume that $\xi_k,\zeta_{k} \to 0$ as $k \to \infty$.
	Then $\hat x(t_k) =H_{\Delta}\xi_k \to 0$ and
	$x(t_k)$, $u_k \to 0$ as $k \to \infty$.
	Since
	\begin{equation*}
	\|x(t_k+\tau)\| \leq
	e^{\|A\| h} \|x(t_k)\| +
	\int^{h}_0 e^{\|A\| h} \|B\| dt \cdot \|u_k\|
	\end{equation*}
	for all $k \in \mathbb{Z}_+$ and all $\tau \in [0,h)$,
	we derive $x(t) \to 0$ ($t \to \infty$).
	Similarly, we see from the dynamics of the estimator $\Sigma_E$ that
	$\hat x(t) \to 0$ as $t \to \infty$.
	This completes the proof.
	\hspace*{\fill} $\Box$
\end{proof}
This proposition allows us to conclude
Problem \ref{problem:clockoffset} can be solved by
finding LTI controllers $\Sigma_C$ achieving $\xi_k$, $\zeta_{k} \to 0$ ($k \to \infty$)
for every $\Delta \in [\underline \Delta, \overline \Delta]$
and for every initial states $\xi_0$ and $\zeta_{0}$.

The following result allows us to conclude that the
discretized system $\Sigma_d$
is detectable and stabilizable for all $\Delta$ and almost all $h$
if the plant $(A,B)$ is stabilizable.
\begin{proposition}
	\label{prop:detectability_stabilizability}
	The discretized system $\Sigma_d$ in 
	\eqref{eq:extended_sys} is detectable for all $\Delta$ and $h$.
	Moreover, $\Sigma_d$ is stabilizable for all $\Delta$ if 
	Assumption \ref{ass:stabilizability} holds.
\end{proposition}
\begin{proof}
	Let us first obtain another realization 
	$(\bar{F}_{\Delta}, \bar{G}_{\Delta}, \bar{H}_{\Delta})$ of 
	the discretized system $\Sigma_d$ in \eqref{eq:extended_sys}.
	We can transform $F_{\Delta}$ into
	\begin{align*}
	\bar{F}_{\Delta} &:=
	T^{-1}F_{\Delta}T =
	\begin{bmatrix}
	\Lambda &\quad  0 \\
	0 &\quad  0
	\end{bmatrix},
	\text{~~where~~}
	T :=
	\begin{bmatrix}
	-\Theta &\quad  -I \\
	I + \Theta &\quad  I
	\end{bmatrix}. 
	\end{align*}
	Furthermore,
	if we define
	\begin{equation*}
	J_1 :=
	\int^{h}_0 e^{-A\tau} d\tau,\qquad
	J_2 :=
	\int^{h-\Delta}_0 e^{-A\tau} d\tau,
	\end{equation*}
	then we obtain
	\begin{align*}
	\bar G_{\Delta} := 
	T^{-1}G_{\Delta} 
	=
	\begin{bmatrix}
	\Lambda J_1B \\
	-\Lambda(J_1 - (I+\Theta)J_2)B- \Theta\Lambda J_1 B
	\end{bmatrix}
	\end{align*}
	and
	$
	\bar{H}_{\Delta} := H_{\Delta}T =
	[
	I + \Theta \quad  I
	].
	$
	We have thus another realization 
	$(\bar{F}_{\Delta}, \bar{G}_{\Delta}, \bar{H}_{\Delta})$
	for $\Sigma_d$.
	
	Next we check detectability and stabilizability
	by using the realization 
	$(\bar{F}_{\Delta}, \bar{G}_{\Delta}, \bar{H}_{\Delta})$. 
	Define
	\begin{equation}
	\label{eqw:L_def}
	L_{\Delta} : = 
	\begin{bmatrix}
	\Lambda (I+\Theta)^{-1} \\ 0
	\end{bmatrix}.
	\end{equation}
	Then we have that
	\begin{equation}
	\label{eq:FLH}
	\bar F_{\Delta}-L_{\Delta} \bar H_{\Delta}
	=
	\begin{bmatrix}
	0 & -\Lambda (I+\Theta)^{-1} \\
	0 & 0
	\end{bmatrix},
	\end{equation}
	and clearly
	$\bar F_{\Delta}-L_{\Delta} \bar H_{\Delta}$ is Schur stable.
	Therefore,
	the discreteized system $\Sigma_d$ is detectable 
	for all $h$ and $\Delta$.
	
	To show stabilizability, we use
	the well-known rank conditions (see, e.g., \cite[Sec.~3.2]{zhou1996}).
	We have that $[zI - \bar{F}_{\Delta}~~  \bar{G}_{\Delta}]$
	is full row rank for all $z \in \mathbb{D}^c$ if and only if
	\begin{equation*}
	\begin{bmatrix}
	zI - \Lambda  &\quad \Lambda J_1B
	\end{bmatrix}		
	=
	\begin{bmatrix}
	zI - e^{Ah}  &\quad \int^{h}_0 e^{A\tau} B d\tau 
	\end{bmatrix}		
	\end{equation*}
	is full row rank for all $z \in \mathbb{D}^c$.
	Hence, the discretized system $\Sigma_d$ is stabilizable for all $\Delta$ if 
	Assumption \ref{ass:stabilizability} holds.
\end{proof}

\subsection{Extension to the output feedback case 
	with disturbances and noise}
Instead of $\Sigma_P$ in \eqref{eq:original_plant}, consider 
a plant $\Sigma_P'$ with disturbances, noise, and output feedback:
\begin{align*}
\Sigma_P':~
\begin{cases}
\dot x(t) = A x(t) + Bu(t) + d(t) \\
y(t) = Cx(t) + n(t),
\end{cases}
\end{align*}
where $d(t) \in \mathbb{R}^{n}$ and $n(t)$, $y(t) \in \mathbb{R}^p$
are the disturbance, measurement noise, and output of the plant, respectively.
As in \cite[Chap.~3]{Garcia2014}, \cite{Xu2005CDC}, 
and the references therein, we assume that
a smart sensor is co-located with the
plant and that the sensor has the following observer to 
generate the state estimate, which is sampled and sent to the controller side:
\begin{equation*}
\Sigma_O:~
\dot{\bar x}(t) = A {\bar x}(t) + Bu(t) + L(y(t)-C\bar x(t)),
\end{equation*}
where $\bar x(t) \in \mathbb{R}^n$ is the state estimate and $L$ is an observer gain such that
$A-LC$ is Hurwitz.
The sampler sends the state estimate $\bar x$, 
and
the resulting dynamics of the time-stamp aware estimator $\Sigma_E'$ is provided by
\begin{equation*}
\Sigma_E':~
\begin{cases}
\dot{\hat x}(t) = A\hat{x}(t) + Bu(t)
&\qquad(t_k \leq t < t_{k+1})	 \\
\hat{x}(\hat s_k) = \bar x(s_k) + w_k,
&\qquad(k \in \mathbb{Z}_+),
\end{cases}
\end{equation*} 
where $w_k \in \mathbb{R}^n$ is the quantization noise. 
A calculation similar to the one performed in the proof of
in Lemma \ref{lem:discretization} can be used to show
that the dynamics of the discretized system $\Sigma_d'$ is given by
\begin{equation}
\label{eq:discre_sys_disturbances}
\Sigma_d':~
\xi_{k+1} = F_{\Delta} \xi_k + G_{\Delta} u_k + 
d_k,\qquad \eta_k = H_{\Delta}\xi_k,
\end{equation}
where $d_k := [
d_{1,k}^{\top}~~ d_{2,k}^{\top}
]^{\top}$ and
\begin{align*}
d_{1,k} &:= \int^{t_{k+1}}_{t_k} e^{A(t_{k+1}-\tau)} d(\tau) d\tau - d_{2,k}\\
d_{2,k} &:= 
-e^{A(t_{k+1}-\hat s_k)}
\Bigg(
e^{(A-LC)(s_k-t_k)}e_k - w_k  \\
&\hspace{30pt}+ \int^{s_k}_{t_k} 
\left(
e^{(A-LC)(s_k-\tau)} (d(\tau)-Ln(\tau)) 
- e^{A(s_k-\tau)} d(\tau) \right) d\tau
\Bigg) \\
e_k &:= x(t_k) - \bar x(t_k). 
\end{align*}
The only difference from
the original idealized system $\Sigma_d$ in \eqref{eq:extended_sys} is
that $\Sigma_d'$ has the disturbance $d_k$. 
Hence, 
for the output feedback case with
bounded disturbances and noise,
solutions of Problem~\ref{problem:clockoffset} achieve
the boundedness of the closed-loop state.
\begin{proposition}
	Assume that $\xi_k, \zeta_{k} \to 0$ as $k \to \infty$
	for the idealized system $\Sigma_d$ in Lemma \ref{lem:discretization}
	(in the context of state feedback without distubances and measurement noise).
	If $d(t)$, $n(t)$, and $w_k$ are bounded for all $t \geq 0$ and 
	all $k \in \mathbb{Z}_+$,
	then the states
	$x(t)$, $\bar x(t)$, $\hat x(t)$, and $\zeta_{k} $ are also bounded for all $t \geq 0$ and all
	$k \in \mathbb{Z}_+$.
	Moreover, if $d(t) = n(t) = w_k = 0$ for all $t \geq 0$ and 
	all $k \in \mathbb{Z}_+$, then
	$x(t)$, $\bar x(t)$, $\hat x(t)$, and $\zeta_{k} $ converge to the origin.
\end{proposition}
\begin{proof}
	Since $d_{k}$
	is bounded for every $k \geq 0$ and every $s_k,\hat s_k \in [t_k,t_{k+1})$, 
	it follows that $\xi_k$ and $\zeta_{k} $ are also bounded
	for all $k \geq 0$.
	The rest of the proof follows the similar lines as that of 
	Proposition~\ref{prop:stability_equivalence}, 
	and hence it is omitted.
\end{proof}
%


See also \cite{Wakaiki2016ACC} for the $\mathcal{L}^2$-gain analysis 
of systems with time-varying offsets.

\section{Controller Design via Simultaneous Stabilization}
\subsection{Preliminaries}
We first consider a general simultaneous stabilization problem 
not limited to the system introduced in Section 2.

The transfer function $P$ 
of the system $\xi_{k+1} = F\xi_k+Gu_k,~y_k = H\xi_k$
is usually defined by the Z-transform of the system's impulse response,
i.e.,
$H(zI-F)^{-1}G$,
but in this paper, we define the transfer function $P$ 
by $P(z) := H(1/z \cdot I-F)^{-1}G$ for 
consistency of the Hardy space theory;
see \cite[Sec. 2.2]{vidyasagar1985} for details.
Hence the transfer function of a causal system is not proper.
We say that $C \in \mathbf{M}(\mathcal{RF}^{\infty})$ 
{\it stabilizes} $P \in \mathbf{M}(\mathcal{RF}^{\infty})$
if $(I+PC)^{-1}$, $C(I+PC)^{-1}$, and $(I+PC)^{-1}P$ belong
to $\mathbf{M}(\mathcal{RH}^{\infty})$. 
We recall that when these three transfer functions belong to 
$\mathbf{M}(\mathcal{RH}^{\infty})$, 
they will have no poles in the closed unit disk.

Consider the family of plants $P_\theta 
\in \mathbf{M}(\mathcal{RF}^{\infty})$ 
parameterized by $\theta\in S$, 
where $S$ is a nonempty parameter set,
and assume that we have a doubly coprime factorization of 
$P_{\theta}$ over $\mathcal{RH}^{\infty}$ 
\begin{equation}
\label{eq:doubly_coprime}
\begin{bmatrix}
Y_{\theta}  &\quad  X_{\theta} \\
-\tilde N_{\theta} &\quad  \tilde D_{\theta}
\end{bmatrix}
\begin{bmatrix}
D_{\theta}  &\quad  -\tilde X_{\theta} \\
N_{\theta} &\quad  \tilde Y_{\theta}
\end{bmatrix} = I,
\end{equation}
where $P_{\theta}=N_{\theta}D_{\theta}^{-1}$ and 
$P_{\theta}=\tilde D_{\theta}^{-1}\tilde N_{\theta}$ are
a right coprime factorization and a left coprime factorization,
respectively.
We explicitly construct the matrices in \eqref{eq:doubly_coprime} using
a stabilizable and detectable realization of $P_{\theta}$;
see, e.g.,
\cite[Theorem 4.2.1]{vidyasagar1985}.

The following theorem provides a necessary and sufficient condition 
for simultaneous stabilization:
\begin{theorem}[\hspace{-0.001pt}\cite{Vidyasagar1982, vidyasagar1985}]
	\label{thm:SS_CM}
	{\it
		Given a nonempty set $S$, consider the plant $P_{\theta}$ 
		having a doubly coprime factorization \eqref{eq:doubly_coprime} 
		for each $\theta \in S$.
		Fix $\theta_0 \in S$ and define
		\begin{equation}
		\label{eq:UVdef}
		\begin{bmatrix}
		U_{\theta} &
		V_{\theta}
		\end{bmatrix}
		:=
		\begin{bmatrix}
		\tilde D_{\theta} &\quad
		\tilde N_{\theta}
		\end{bmatrix}
		\begin{bmatrix}
		\tilde Y_{\theta_0} & \quad -N_{\theta_0}\\
		\tilde X_{\theta_0} & \quad D_{\theta_0}
		\end{bmatrix}
		\qquad
		(\theta \in S).
		\end{equation}
		Then $(V_{\theta}, U_{\theta})$ is right coprime for every
		$\theta \in S$. Moreover, 
		there exists a controller that stabilizes 
		$P_{\theta}$ for every $\theta \in S$ if and only if
		there exists $Q \in \mathbf{M}(\mathcal{RH}^{\infty})$ such that
		for all $\theta \in S$,
		\begin{equation}
		\label{eq:SS_Unimodular}
		(U_{\theta} + V_{\theta}Q)^{-1} \in \mathbf{M}(\mathcal{RH}^{\infty}).
		\end{equation}
		Such a stabilizing controller is given by
		\begin{equation}
		\label{eq:SScontroller}
		C :=  (\tilde X_{\theta_0} + 
		D_{\theta_0} Q)
		(\tilde Y_{\theta_0} - N_{\theta_0}Q )^{-1}.
		\end{equation}
	}
\end{theorem}

\begin{remark}
		Although the simultaneous stabilization of a {\em finite} family of plants is considered in
		\cite[Sec. 5.4]{vidyasagar1985} and 
		\cite{Vidyasagar1982},
		generalization to an {\em arbitrary} family of plants is readily apparent,
		as mentioned in the last paragraph of Section 3 in \cite{Vidyasagar1982}.
\end{remark}
\begin{remark}
		A {\em left} coprime factorization of stabilizing controllers
		is used in 
		\cite[Sec. 5.4]{vidyasagar1985} and 
		\cite{Vidyasagar1982},
		whereas we represent controllers by a {\em right} coprime 
		factorization in \eqref{eq:SScontroller}.
		Therefore, Theorem \ref{thm:SS_CM} is slightly different
		from its counterpart in 
		\cite[Sec. 5.4]{vidyasagar1985} and
		\cite{Vidyasagar1982}.
\end{remark}

\subsection{Robust Controller Design}
It is generally not easy to verify 
in a computationally efficient fashion
that
a transfer function $Q$ satisfying 
\eqref{eq:SS_Unimodular} exists.
In the next theorem, we develop a simple sufficient condition for
\eqref{eq:SS_Unimodular} to hold, by exploiting geometric properties
on $\mathcal{H}^{\infty}$ inspired by results
on strong stabilization \cite{zeren2000}.
\begin{theorem}
	\label{thm:suff}
	{\it
		Given a nonempty set $S$, 
		assume that each plant $P_{\theta}$ ($\theta \in S$)
		has a doubly coprime factorization \eqref{eq:doubly_coprime} 
		such that 
		there exist $\theta_0 \in S$, 
		$W \in \mathbf{M}(\mathcal{RH}^{\infty})$,
		and $R(\theta) \in \mathbf{M}(\mathbb{R})$ satisfying 
		$\tilde D_{\theta} = \tilde D_{\theta_0}$ and
		\begin{equation}
		\label{eq:N_cond}
		\tilde N_{\theta}(z) - \tilde N_{\theta_0}(z) = R(\theta) W(z),
		\end{equation}
		for all $\theta \in S$.
		If there exists $Q\in \mathbf{M}(\mathcal{RH}^{\infty})$ 
		satisfying
		the following $\mathcal{H}^{\infty}$-norm condition:
		\begin{equation}
		\label{eq:Nehari_problem}
		\|W (\tilde X_{\theta_0} + D_{\theta_0} Q)\|_{\infty} < 
		\frac{1}{\sup_{\theta \in S}\|R(\theta)\|},
		\end{equation}
		then $Q$ satisfies \eqref{eq:SS_Unimodular}, and hence
		the controller $C$ in \eqref{eq:SScontroller} stabilizes 
		$P_{\theta}$ for every $\theta \in S$.
	}
\end{theorem}
\begin{proof}
	We define $U_{\theta}$ and $V_{\theta}$ as in
	\eqref{eq:UVdef}. Since $\tilde D_\theta = \tilde D_{\theta_0}$, 
	it follows from \eqref{eq:UVdef} and
	the Bezout identity $\tilde D_{\theta_0}
	\tilde Y_{\theta_0}+\tilde N_{\theta_0}\tilde X_{\theta_0} = I$ 
	in \eqref{eq:doubly_coprime} that
	\begin{align*}
	U_{\theta} 
	&= \tilde D_{\theta} \tilde Y_{\theta_0} 
	+ \tilde N_{\theta} \tilde X_{\theta_0} 
	= I + (\tilde N_{\theta} - \tilde N_{\theta_0}) \tilde X_{\theta_0}.
	\end{align*}
	Moreover, since $\tilde D_{\theta_0}^{-1} \tilde N_{\theta_0} 
	= N_{\theta_0} D_{\theta_0}^{-1}$, we obtain
	\begin{align*}
	V_{\theta}
	&= -\tilde D_{\theta} N_{\theta_0} + \tilde N_{\theta} D_{\theta_0} 
	= (\tilde N_{\theta} - \tilde N_{\theta_0}) D_{\theta_0}.
	\end{align*}
	Hence 
	$
	U_{\theta} + V_{\theta}Q
	=
	I +  
	(\tilde N_{\theta} - \tilde N_{\theta_0})
	(\tilde X_{\theta_0} + D_{\theta_0}Q).
	$
	Since $(I+\Phi)^{-1} \in \bf{M}(\mathcal{RH}^{\infty})$ for all
	$\Phi \in \bf{M}(\mathcal{RH}^{\infty})$ satisfying $\|\Phi\|_{\infty} < 1$,
	it follows that
	if
	\begin{equation}
	\label{eq:Q_N_norm}
	\|	(\tilde N_{\theta} - \tilde N_{\theta_0})
	(\tilde X_{\theta_0} + D_{\theta_0}Q)\|_{\infty} < 1
	\qquad (\theta \in S),
	\end{equation}
	then \eqref{eq:SS_Unimodular} holds for all $\theta \in S$. 
	From the assumption \eqref{eq:N_cond}, 
	\begin{align*}
	&\|	(\tilde N_{\theta} - \tilde N_{\theta_0})
	(\tilde X_{\theta_0} + D_{\theta_0}Q) \|_{\infty} 
	\leq 
	\|R(\theta)\| \cdot
	\|	W
	(\tilde X_{\theta_0} + D_{\theta_0}Q) \|_{\infty}.
	\end{align*}
	Hence if $Q$ satisfies \eqref{eq:Nehari_problem} 
	for all $\theta \in S$,
	then 
	\eqref{eq:Q_N_norm} holds, and consequently
	$P_{\theta}$
	is simultaneously stabilizable by $C$ in \eqref{eq:SScontroller}
	from Theorem \ref{thm:SS_CM}.
	\hspace*{\fill} $\Box$ 
\end{proof}

The proposition below shows that our discretized system $\Sigma_d$ in \eqref{eq:extended_sys}
always satisfies
the assumptions on $\tilde D_{\theta}$ and $\tilde N_{\theta}$
that appear in Theorem \ref{thm:suff}. 
This result also provides the matrices $R$ and $W$ in \eqref{eq:N_cond}
without explicitly calculating a coprime factorization of $P_{\theta}$ 
for all $\theta \in S$.
\begin{proposition}
	\label{prop:DN_cond}
	Define the transfer function $P_{\Delta}(z) := 
	H_{\Delta} (1/z\cdot I - F_{\Delta})^{-1} G_{\Delta}$.
	For all $\Delta \in (-h,h)$,
	there exists a doubly coprime factorization \eqref{eq:doubly_coprime}
	such that $\tilde D_{\Delta}(z) = \tilde D_{0}(z) = I - ze^{Ah}$,
	and 
	\eqref{eq:N_cond} holds with
	\begin{align}
	R(\Delta) &:= \int^{\Delta}_0 e^{A(h-\tau)}B d\tau 
	\in \mathbb{R}^{n \times m} \label{eq:R_def}\\
	W(z) &:= z(z-1) \in \mathcal{RH}^{\infty}. \label{eq:W_def}
	\end{align}
\end{proposition}
\begin{proof}
	Consider the realization $(\bar F_{\Delta}, \bar G_{\Delta}, \bar H_{\Delta})$ in
	the proof of 
	Proposition \ref{prop:detectability_stabilizability}.
	For every $\Delta \in 
	(-h, h)$,
	the matrix $L_{\Delta}$ in \eqref{eqw:L_def}
	achieves the Schur stability of
	$\bar F_{\Delta}-L_{\Delta} \bar H_{\Delta}$ as shown in \eqref{eq:FLH}.
	From the realization of $\tilde D_{\Delta}$, e.g., in
	\cite[Theorem 4.2.1]{vidyasagar1985},
	we can write
	$\tilde D_{\Delta}$ as
	\begin{equation*}
	\tilde D_{\Delta}(z) = I - \bar H_{\Delta}
	(1/z\cdot I - (\bar F_{\Delta}-L_{\Delta} \bar H_{\Delta}))^{-1}L_{\Delta}
	=
	I - z\Lambda.
	\end{equation*}
	Noticing that the far right-hand side of the equation above
	does not depend on $\Delta$,
	we have $\tilde D_{\Delta}(z) = \tilde D_{0}(z) = I - z\Lambda$.
	
	It follows that $\tilde N_{\Delta} - \tilde N_{0} = 
	\tilde D_0(P_{\Delta} - P_0)$. 
	From the realization $(\bar{F}_{\Delta}, \bar{G}_{\Delta}, \bar{H}_{\Delta})$,
	we see that
	\begin{align*}
	P_{\Delta}(z) &=
	z \bigl((I+\Theta) (I - z\Lambda)^{-1}\Lambda J_1 
	- \Lambda (J_1 - (I+\Theta)J_2)  - \Theta \Lambda J_1\bigr)B.
	\end{align*}
	Since $\Theta = 0$ and $J_2 = J_1$ for $\Delta = 0$,
	it follows that
	\begin{align}
	P_{\Delta}(z) - P_0(z) &=
	z\bigr( 
	\Theta (I-z\Lambda)^{-1}\Lambda J_1 
	-\Lambda (J_1 - (I+\Theta)J_2) - \Theta\Lambda J_1
	\bigr)
	B.
	\label{eq:PD-P0_first}
	\end{align}
	
	On the other hand,
	we have
	$
	J_1 - (I+\Theta)J_2 = 
	\int^{\Delta}_0 e^{-A\tau} d\tau =: \bar \Theta,	
	$
	and
	\begin{align}
	\label{eq:int_exp}
	\bar \Theta A = 
	-\int^{\Delta}_0 \left(\frac{d}{d\tau} e^{-A\tau}\right) d\tau
	= -\Theta.
	\end{align}
	Since $A(I-z\Lambda) = A-zAe^{Ah} = A-ze^{Ah}A = (I-z\Lambda)A$,
	it follows that $A(I-z\Lambda)^{-1} = (I-z\Lambda)^{-1}A$.
	Therefore we derive from \eqref{eq:PD-P0_first}
	\begin{align*}
	P_{\Delta} - P_0 &=
	z\Lambda \bar \Theta
	(I-z\Lambda)^{-1}
	\left(
	-I + 
	z\Lambda(I-AJ_1)
	\right)B.
	\end{align*}
	Similarly to \eqref{eq:int_exp}, we have
	$I-AJ_1 = \Lambda^{-1}$,
	and hence
	$P_{\Delta} - P_0 =
	z(z-1)\Lambda \bar \Theta (I-z\Lambda) ^{-1}B.
	$
	Since $\lambda$, $\bar\Theta$, and $(I-z\Lambda)^{-1}$
	are commutative, we derive 
	\begin{align*}
	P_{\Delta} - P_0 =
	z(z-1)
	(I-z\Lambda) ^{-1}
	\bar \Theta\Lambda  B,
	\end{align*}
	and 
	$
	\tilde N_{\Delta} - \tilde N_0
	=
	z(z-1)
	\bar \Theta \Lambda B.
	$
	Thus \eqref{eq:N_cond} holds with
	$R$ in \eqref{eq:R_def} and
	$W$ in \eqref{eq:W_def}
	\hspace*{\fill} $\Box$ 
\end{proof}

Define 
\begin{equation}
\label{eq:gamma}
\gamma := \frac{1}{{\displaystyle　\max_{\Delta \in [\underline \Delta, \overline \Delta]}} \left\|\int^{\Delta}_0
	e^{A(h-\tau)}B d\tau\right\|}.
\end{equation}
From Theorem \ref{thm:suff}, to obtain a controller $\Sigma_C$ as in  \eqref{eq:controller},
it is enough to solve the following suboptimal problem:
Find $Q \in \mathbf{M}(\mathcal{RH}^{\infty})$ satisfying
$\|W(\tilde X_0 - D_0 Q)\|_{\infty} < \gamma$.
This problem is equivalent to a standard suboptimal $\mathcal{H}^{\infty}$ control problem
\cite[Chaps.~16,~17]{zhou1996}:
Find $Q \in \mathbf{M}(\mathcal{RH}^{\infty})$ such that 
$\|\mathcal{F}_{\ell}(\Phi,Q)\|_{\infty} < \gamma$, where $\Phi$ is defined by
\begin{equation}
\label{eq:G_def}
\Phi:=
\begin{bmatrix}
W\tilde X_0  &\quad WD_0  \\
-I & 0
\end{bmatrix}.
\end{equation}
The results of this section can be summarized through the following controller design algorithm:
\begin{algorithm}
	\label{alg:controller_design}
	\begin{enumerate}
		\item
		Using the realization 
		\[
		\hspace{-22pt}(\bar F_0, \bar G_0, \bar H_0) = 
		\left(
		\begin{bmatrix}
		e^{Ah} &\quad 0 \\
		0 & \quad 0
		\end{bmatrix}\!,~
		\begin{bmatrix}
		\int^h_0e^{Ah}Bd\tau  \\
		0 
		\end{bmatrix}\!,~
		\begin{bmatrix}
		I &\quad I
		\end{bmatrix}
		\right),
		\]
		the matrix 
		$			L_0 = 
		[
		e^{A^{\top}h} ~~ 0
		]^{\top},
		$		
		and an arbitrary matrix $K_0$ such that $\Phi := \bar F_0 - \bar G_0 K_0$ is Schur stable,
		set
		\begin{align*}
		D_{0}(z) &:=
		I - K_0(1/z\cdot I -\Phi )^{-1} \bar G_{0} \\
		N_{0}(z) &:=
		\bar H_0(1/z\cdot I - \Phi )^{-1} \bar G_{0} \\
		\tilde X_{0}(z) &:= K_{0}
		(1/z\cdot I - \Phi )^{-1}  L_{0} \\
		\tilde Y_{0}(z) &:= I + H_{0}
		(1/z\cdot I - \Phi )^{-1} K_{0} \\
		W(z) &:= z(z-1).
		\end{align*}
		
		\item 
		For a given offset interval $[\underline \Delta, \overline \Delta]$,
		set $\gamma$ as in \eqref{eq:gamma}, and
		solve the $\mathcal{H}^{\infty}$ control problem \cite[Chaps.~16,~17]{zhou1996}:
		Find $Q \in \mathbf{M}(\mathcal{RH}^{\infty})$ such that 
		$\|\mathcal{F}_{\ell}(\Phi,Q)\|_{\infty} < \gamma$,
		where $\Phi$ is defined by \eqref{eq:G_def}.

		\item 
		If the $\mathcal{H}^{\infty}$ control problem is not solvable,
		then the algorithm fails.
		Otherwise the transfer function $C$ of the 
		controller $\Sigma_C$
		is given by $C = (\tilde X_{0} + 
		D_{_0} Q)
		(\tilde Y_{_0} - N_{0}Q )^{-1}$.
	\end{enumerate}
\end{algorithm}

\begin{remark}
		We have from Proposition \ref{prop:DN_cond}
		that $P_{\Delta} = P_0 + WD_0^{-1}R(\Delta)$ for constant $\Delta$, where
		$P_\Delta$ is expressed as the nominal component $P_0$ 
		plus the uncertainty block $WD_0^{-1}R(\Delta)$.
		If we obtain a similar formula for the case of time-varying offsets as
		studied for systems with aperiodic sampling in \cite{Fujioka2009},
		we can deal with the stabilization problem of systems with time-varying offsets
		through a small gain theorem. Although
		the uncertainty part of the discretized system $\Sigma_d$ may be non-causal, 
		the small gain theorem for systems with non-causal uncertainty in \cite{Unal2008Automatica}
		can be used.
		This extension is a subject for future research.
\end{remark}

\begin{example}
Consider the unstable batch reactor studied 
in \cite{Rosenbrock1972}, where
the system matrices $A$ and $B$ in \eqref{eq:original_plant}
are given by
\begin{align*}
A &:= 
\begin{bmatrix}
1.38 &\quad -0.2077 &\quad 6.715 &\quad -5.676 \\
-0.5814 &\quad -4.29 &\quad 0 &\quad 0.675 \\
1.067 &\quad 4.273 &\quad -6.654 &\quad 5.893 \\
0.048 &\quad 4.273 &\quad 1.343 &\quad -2.104
\end{bmatrix},\qquad 
B :=
\begin{bmatrix}
0 &\quad 0 \\
5.679 &\quad 0 \\
1.136 &\quad-3.146 \\
1.136 &\quad 0 
\end{bmatrix}.
\end{align*}
This example has been developed over the years as a benchmark example for networked
control systems, and
its data were transformed by a change of basis
and time scale \cite{Rosenbrock1972}.

Here we compare the proposed method with the robust stabilization method in
\cite{doyle1981} and \cite[Chap.~7]{vidyasagar1985} based on
the following fact: {\em
	Consider a family of plants $P_{\Delta} \in \mathbf{M}(
	\mathcal{RF}^{\infty})$ with $\Delta \in [\underline \Delta, \overline \Delta]$.
	Assume 
	that 
	$P_{\Delta}$ has no poles on $\mathbb{T}$ and the same number of unstable poles for every $\Delta \in [\underline \Delta, \overline \Delta]$
	and
	that a function $r \in \mathcal{RH}^{\infty}$ satisfies
	\begin{equation}
	\label{eq:r_def}
	\|P_{\Delta}(e^{j\omega}) - P_{0}(e^{j\omega}) \| < |r(e^{j\omega})|
	\end{equation} 
	for all $\Delta \in [\underline \Delta, \overline \Delta]$
	and all $\omega \in [0,2\pi]$.
	If the controller $C \in \mathbf{M}(
	\mathcal{RF}^{\infty})$ 
	stabilizes $P_0$ and
	satisfies 
	\begin{equation}
	\label{eq:suff_cond_conv}
	\left\|
	r  C (I + P_0C)^{-1}
	\right\|_{\infty} \leq 1,
	\end{equation}
	then $C$ stabilizes $P_{\Delta}$ for all $\Delta \in [\underline \Delta, \overline \Delta]$.}
The order of such a controller is typically equal to
the order of the following transfer function:
\[\begin{bmatrix}
0 &\quad  rI \\ I &\quad  P_0
\end{bmatrix}.\]

We compute the length of the allowable offset interval $[\underline \Delta, \overline \Delta]$
obtained from the sufficient condition \eqref{eq:Nehari_problem} for each $h \in [0.2,~3.6]$,
which is shown as the solid line in Fig.~\ref{fig:MIMO_ex}.
On the other hand, the dashed line in the figure represents 
the length of the offset interval $[\underline \Delta, \overline \Delta]$
obtained from the robust control approach that leads to the condition
\eqref{eq:suff_cond_conv} with an appropriate function $r \in \mathcal{RH}^{\infty}$
satisfying \eqref{eq:r_def}.
For example, we use $r(z) = 0.1766(z-1)/(z-0.9389)$ for $h=1$ and $
[\underline \Delta, \overline \Delta] = 
[-0.02, 0.02]$, and this $r$ satisfies \eqref{eq:r_def} and
\[
\frac{1}{2\pi} \int_{0}^{2\pi} 
\left(
|r(e^{j\omega})| - \|P_{\Delta}(e^{j\omega}) - P_{0}(e^{j\omega}) \| 
\right)
\leq 8.5 \times 10^{-3}
\]
for all $\Delta \in  [\underline \Delta, \overline \Delta] = 
[-0.02, 0.02]$.
The solid line is obtained by finding the maximum and minimum of
$\Delta$ that satisfies
the condition
\begin{equation*}
\left\|\int^{\Delta}_0
e^{A(h-\tau)}B d\tau\right\| \leq
\frac{1}{\min_{Q\in \mathcal{RH}^{\infty}} \|\mathcal{F}_{\ell} (\Phi,Q)\|_{\infty}},
\end{equation*}
whereas to derive the dashed line, we first calculate $r$ satisfying
\eqref{eq:r_def} for a fixed $[\underline \Delta, \overline \Delta]$ and
then check the existence of a controller $C$ that stabilizes $P_0$ and achieves the
$\mathcal{H}^{\infty}$-norm condition \eqref{eq:suff_cond_conv}.
We see from Fig.~\ref{fig:MIMO_ex}
that the proposed sufficient condition \eqref{eq:Nehari_problem} 
is less conservative
than \eqref{eq:suff_cond_conv}.

Consider the case $h = 1$, and let  $C_1$ and $C_2$ be
controllers that are obtained from the sufficient conditions 
\eqref{eq:Nehari_problem}  and  \eqref{eq:suff_cond_conv} with
the maximum offset length, respectively.
The order of the controller $C_1$ 
is 7, but 
applying balanced model truncation \cite[Chap.~6]{zhou1996}
to the controller $C_{1}$,
we can obtain
an approximated controller $C_{\text{app}}$ with order 5,
which 
satisfies 
$\|C_{\text{app}} - C_{1}\|_{\infty} / \|C_1\|_{\infty} = 0.023$.
From
iterative calculations of the eigenvalues of the 
discretized closed-loop system for each $\Delta$,
we find that both $C_{1}$ and $C_{\text{app}}$
stabilize the discretized system $\Sigma_d$ in \eqref{eq:extended_sys} 
for all $\Delta \in (-1,1)$.
The controller $C_2$
has order $5$ and
allows the offsets $\Delta \in [-0.054,0.068]$ without compromising the closed-loop stability.
Approximated controllers with any order obtained by applying balanced 
model truncation to $C_2$
do not achieve the closed-loop stability even in the case $\Delta = 0$.
For comparison, a linear quadratic regulator whose state weighting matrix  and
input weighting matrix are identity matrices with appropriate dimension stabilizes
the discretized system $\Sigma_d$ in \eqref{eq:extended_sys} 
only for $\Delta \in [-0.029,0.062]$.
From this numerical result, we see that
the derived controller achieves better robust performance against clock offsets 
than
a linear quadratic regulator designed without regard to the clock offset and
also than the robust controller based on \eqref{eq:suff_cond_conv}.

\begin{figure}[tb]
	\centering
	\includegraphics[width = 8cm,clip]{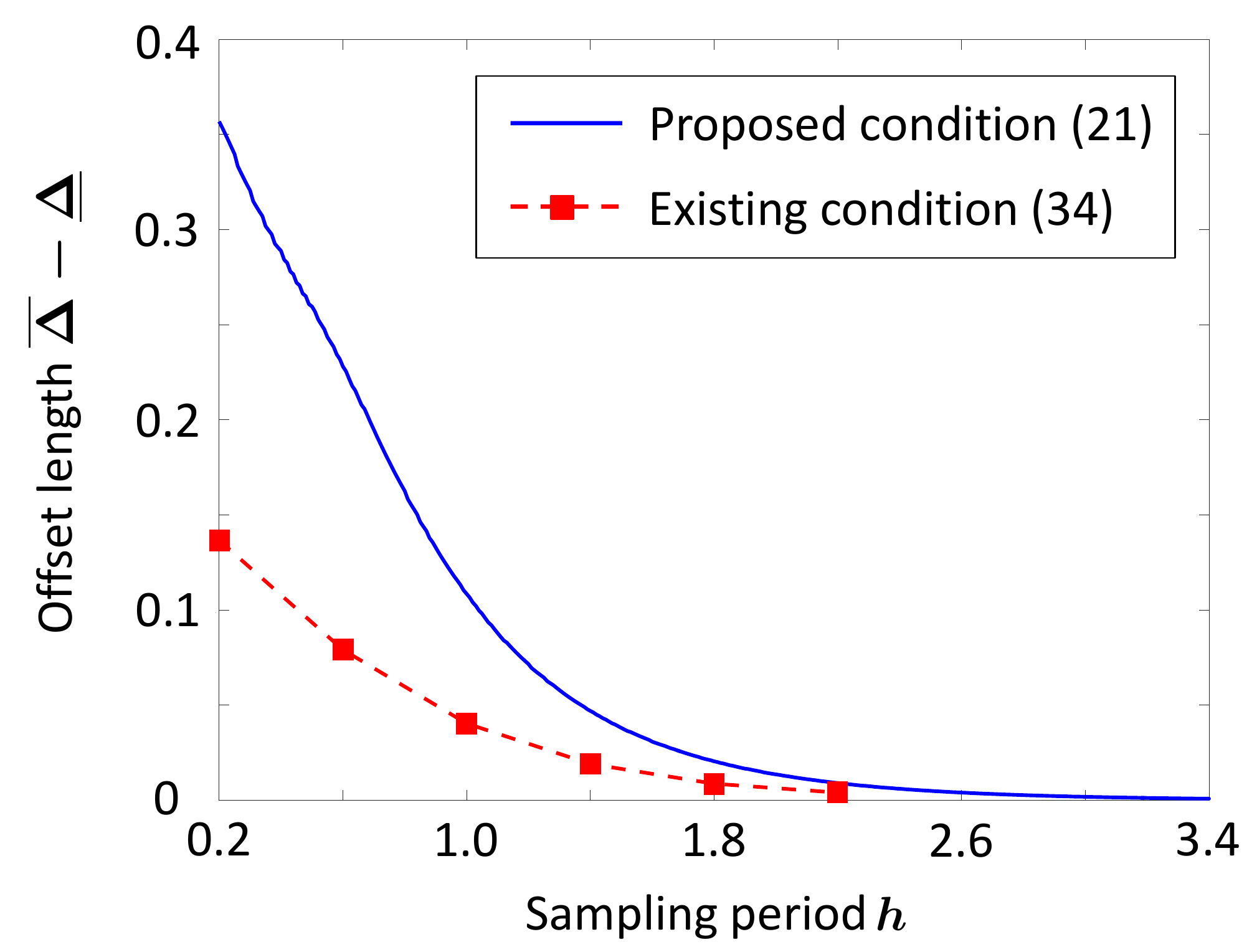}
	\caption{Allowable offset length $\overline \Delta - \underline \Delta$ versus sampling period $h$.}
	\label{fig:MIMO_ex}
\end{figure}	
\end{example}

\section{Exact Bound on Offsets for First-order Systems}
In this section, the bounds
on the clock offset
that were obtained for LTI controllers 
(from Theorems \ref{thm:scalar_LTIcontroller} 
and \ref{thm:suff}) are compared with
the exact bound that would be allowed 
if we restricted our attention to a time-invariant 
static output feedback controller, and with
an offset range that gives a sufficient condition 
for the existence of a time-varying 
2-periodic static output feedback controller.
We also derive a bound obtained using 
standard robust control tools, regarding the clock offset as an
additive uncertainty.

In this section, the plant class is restricted to scalar systems, 
and we reduce the necessary and sufficient condition for stabilizability
in Theorem \ref{thm:SS_CM} to a computationally verifiable one,
which gives an explicit formula for the exact bound on
the clock offset that LTI controllers can allow.

Consider an unstable scalar plant:
$
\dot x = a x + bu
$ 
with $a > 0$.
If $a < 0$, the stabilization problem is trivial because a
zero control input $u(t) = 0$ leads to the stability of the closed-loop system.
So in the reminder o this section, 
we will focus our attention to the case $a > 0$.
The case $a=0$ will be addressed separately later.

Solving explicitly the integrals that appear in \eqref{eq:FGH_def},
the extended system \eqref{eq:extended_sys} is given by
\begin{align}
\xi_{k+1}
&=
\begin{bmatrix}
-\lambda \theta & -\lambda\theta \\
\lambda(1+\theta) & \lambda(1+\theta)
\end{bmatrix} \xi_k
+
\frac{b}{a}
\begin{bmatrix}
-\lambda \theta \\
\lambda(1+\theta) - 1
\end{bmatrix}
u_k \notag \\
y_k &= 
\begin{bmatrix}
0 & 1
\end{bmatrix}
\xi_k,
\label{eq:scalar_sys}
\end{align}
where
$
\lambda := e^{ah}
$ and
$
\theta := e^{-a\Delta} - 1.
$
In what follows, we take $b/a = 1$ for simplicity of notation,
because stabilizability does not depend on the value of this ratio.

The extended system \eqref{eq:scalar_sys}
is stabilizable and detectable except for $\theta = -1$, at which
point the system loses detectability. 
Since $\theta = e^{-a\Delta} - 1$, 
it follows that 
\begin{equation}
\label{eq:Delta_theta}
\Delta \in [\underline{\Delta}, \overline{\Delta}] ~~
\Leftrightarrow
~~
\theta \in [e^{-a\overline{\Delta}}-1, e^{-a\underline{\Delta}}-1] =: S.
\end{equation}
We have from Assumption \ref{ass:sample} 
that $-h < \Delta < h$, and hence
the set $S$ on $\theta$ is a subset of
\begin{equation}
\label{eq:Imax_def}
(e^{-ah} - 1,e^{ah} - 1) =: S_{\max}.
\end{equation}

As in Section 3,
taking the Z-transform of
\eqref{eq:scalar_sys} and then
mapping $z \mapsto 1/z$, we obtain the transfer function $P_{\theta}$:
\begin{equation}
\label{eq:scalar_TF}
P_{\theta} (z) 
=
\frac{(\lambda -1)z}{1-\lambda z}
-
\theta
\frac{\lambda z(z-1)}{1-\lambda z}\quad
(\theta \in S \subset S_{\max}).
\end{equation}
The system \eqref{eq:scalar_TF}
belongs to a class of the so-called interval systems.
The stabilization of general interval systems
has been studied, e.g., in
\cite{Ghosh1988, Olbrot1994}.
Here we shall develop a new approach based on
Theorem \ref{thm:SS_CM}.

\subsection{Main result for scalar plants}
The following theorem gives the exact bound on the 
clock offset for scalar systems:
\begin{theorem}
	\label{thm:scalar_LTIcontroller}
	{\it
		Define $\lambda := e^{ah}$ and $\theta := e^{-a\Delta} - 1$.
		Let $\underline \theta < 0 < \overline \theta$ and 
		consider the set $S$ in \eqref{eq:Delta_theta} of the form
		$S = [\underline \theta, \overline \theta] \subset S_{\max}$.
		There exists a controller that stabilizes 
		$P_{\theta}$ in \eqref{eq:scalar_TF} 
		for all $\theta \in S$, that is,
		there exists $Q \in \mathcal{RH}^{\infty}$ satisfying
		\eqref{eq:SS_Unimodular} for all $\theta \in S$
		if and only if
		\begin{equation}
		\label{eq:M1M2cond}
		(\lambda-1)^2\overline \theta - (\lambda+1)^2\underline \theta < 4\lambda.
		\end{equation}
		In particular, if $-\underline \theta=\overline \theta$, then \eqref{eq:M1M2cond}
		is equivalent to
		\begin{equation}
		\label{eq:LTI_bound_M}
		\overline \theta < \frac{2\lambda}{\lambda^2+1}.
		\end{equation}
		
		Furthermore, 
		define a conformal mapping $\phi$ from
		$\mathbb{G} := \mathbb{C} \setminus
		\{(-\infty, 1/\underline \theta] \cup [1/\overline \theta,\infty) \}$
		to 
		$\mathbb{D}$ by
		\begin{align}
		\label{eq:conformal_map}
		\phi:~ \mathbb{G} \to \mathbb{D}: ~s \mapsto 
		\frac{ 1 -
			\sqrt{(1-\overline \theta s) / (1-\underline \theta s)}
		}
		{ 1 +
			\sqrt{
				(1-\overline \theta s) / (1-\underline \theta s)}}.
		\end{align}
		If \eqref{eq:M1M2cond} holds, then
		a finite-dimensional stabilizing controller $C$ is given by 
		\begin{equation*}
		C := \frac{X_0+D_0Q}{Y_0-N_0Q},\qquad
		Q := \frac{\phi^{-1} \circ g - T_1}{T_2},
		\end{equation*}
		where the $\mathcal{RH}^{\infty}$ functions 
		$N_0$, $D_0$, $X_0$, $Y_0$, $T_1$, and $T_2$ are defined by
		\begin{gather}
		N_{0}(z) := \frac{(\lambda-1)z}{z-c},\qquad
		D_{0}(z) := \frac{1-\lambda z}{z-c} \notag \\
		X_0(z) := \frac{1 -c \lambda}{\lambda - 1},\qquad
		Y_0(z) := -c \label{eq:XY_def} \\
		T_1 (z):=
		\lambda\frac{1-c \lambda}{\lambda-1}
		\frac{z(z-1)}{z-c},~
		T_2 (z):= 
		\lambda \frac{z(z-1)(1-\lambda z)}
		{(z-c)^2}. \notag 
		\end{gather}
		for any arbitrarily fixed $c \in \mathbb{C}$ with $|c| > 1$, and
		any rational function 
		$g:~\bar{\mathbb{D}} \to \mathbb{D}$ that satisfies
		the interpolation conditions $g(0) = \phi(0)$,
		$g(1) = \phi(0)$, and $g(1/\lambda) = \phi(-1)$.
		Such a function $g$ always exists if \eqref{eq:M1M2cond} holds.
	}
\end{theorem}
\begin{proof}
	See Section 4.2.
\end{proof}

\begin{remark}
	The rational function $g$ in Theorem \ref{thm:scalar_LTIcontroller}
	can be obtained from the Schur-Nevanlinna algorithm; see, e.g., 
	\cite{luxemburg2010, wakaiki2012}.
\end{remark}	

\begin{remark}
	\label{rem:related_work}
	Since the inverse mapping $\phi^{-1}$ is given by
	the following rational function:
	\begin{equation}
	\label{eq:phi_inv}
	\phi^{-1}(s) = \frac{4s}{\overline \theta (s+1)^2 
		- \underline \theta (s-1)^2},
	\end{equation}
	the stabilizing controller $C$ is finite dimensional for 
	a rational function $g$.
\end{remark}

If we change the offset variable from $\theta = e^{-a\Delta} - 1$ to $\Delta$,
then \eqref{eq:Delta_theta} and \eqref{eq:M1M2cond} give
the maximum 
length of the offset interval $[\underline \Delta, \overline \Delta]$ allowed 
by an LTI controller.
\begin{corollary}
	\label{coro:maximum_length}
	Assume $-h < \underline \Delta < 0 < \overline \Delta < h$.
	There exists a controller that stabilizes 
	the extended system \eqref{eq:scalar_sys}
	for all $\Delta \in [\underline \Delta, \overline \Delta]$ if and only if 
	\begin{equation}
	\label{eq:Delt_bound_scalar}
	\overline \Delta - \underline \Delta < 
	\frac{2	\left( \log(e^{ah}+1) - \log(e^{ah}-1) \right)}{a}.
	\end{equation}
\end{corollary}
\begin{proof}
	See Section 4.2.
\end{proof}

\begin{remark}	
	\label{rem:a_equal_zero}	
	In the case $a=0$, the extended system $P_{\Delta}$ is given by
	\begin{equation*}
	P_{\Delta}(z) = 
	\frac{hz}{1-z} - \Delta z.
	\end{equation*}
	Similarly to the case $a>0$, one can show that
	there exists a controller stabilizing $P_{\Delta}$
	for all $\Delta \in (-h,h)$.
	This result is consistent with that in the case when
	$a \to 0$ in
	Corollary \ref{coro:maximum_length}, but
	we omit the proof for brevity.
\end{remark}

\begin{example}
	Consider a scalar plant with unstable pole $a = 1$. 
	In Fig.~\ref{fig:h_dependency}, 
	we plot the maximum length of the offset interval $(
	\underline \Delta, \overline \Delta)$
	versus the ZOH-update period $h$. 
	The solid line is the maximum length $\overline \Delta - \underline \Delta$
	and 
	the vertical dotted lines indicate $h = h_0 := (\log(1+\sqrt{2}))/a$.
	If $h < h_0$, then the restriction $-h < \underline \Delta < \overline \Delta
	< h$ arising from Assumption \ref{ass:sample} gives the bound
	$\overline \Delta - \underline \Delta < 2h$.
	On the other hand, if $h \geq h_0$, then
	$\overline \Delta - \underline \Delta$ is bounded by \eqref{eq:Delt_bound_scalar}.
	The maximum offset length $\overline \Delta - \underline \Delta$
	exponentially decreases as $h \geq h_0$ becomes larger. 
	\begin{figure}[bt]
		\centering
		\includegraphics[width = 9.5cm]{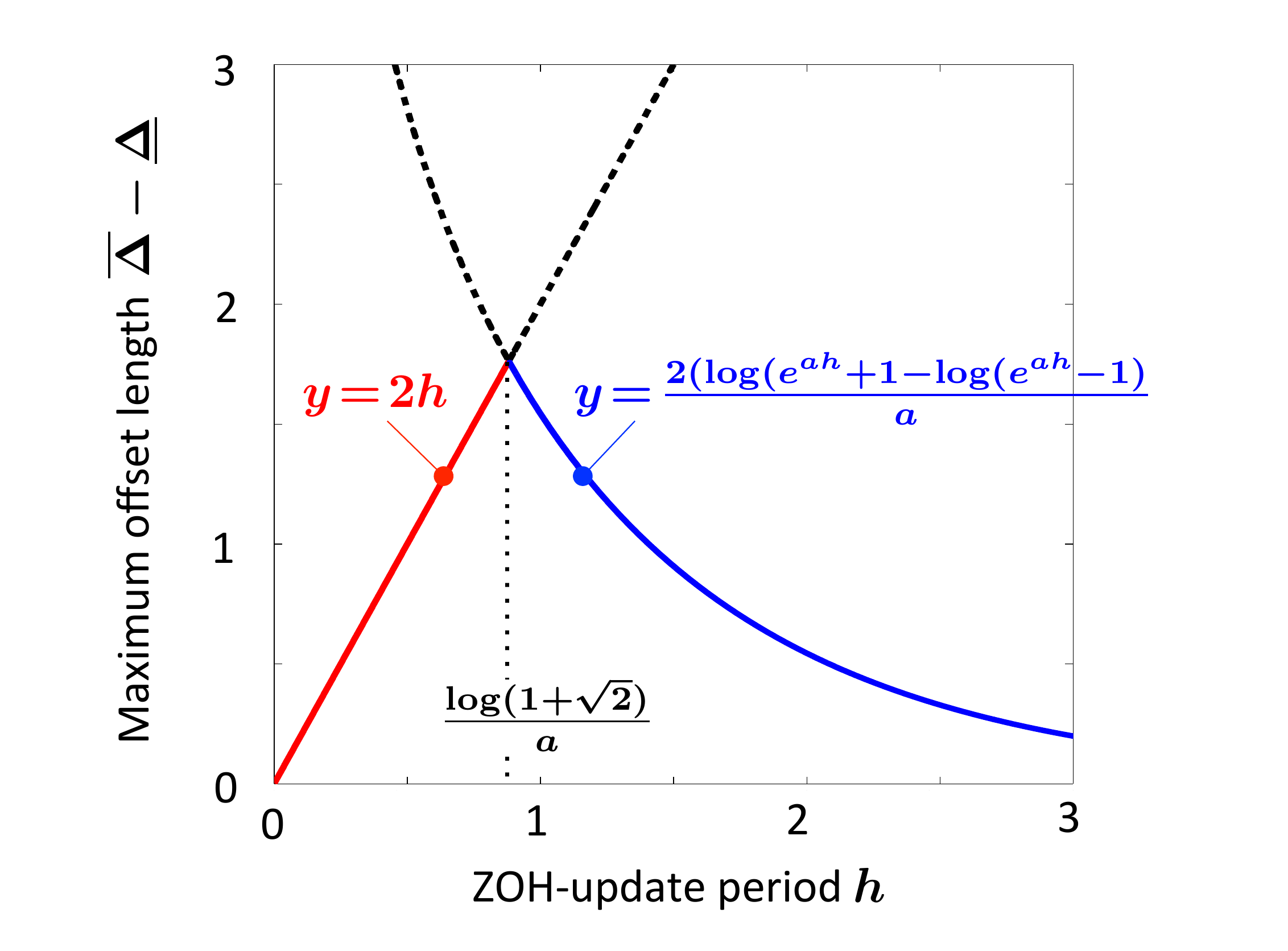}
		\caption{Maximum offset length $\overline \Delta - 
			\underline \Delta$ 
			versus ZOH-update period $h$ (unstable pole $a=1$):
			The solid line is 
			the maximum offset length. The vertical dotted line
			indicates $h= (\log(1+\sqrt{2}))/a$.}
		\label{fig:h_dependency}
	\end{figure}
\end{example}

\subsection{Proofs of Theorem \ref{thm:scalar_LTIcontroller}
	and Corollary \ref{coro:maximum_length}}
We prove Theorem \ref{thm:scalar_LTIcontroller}
by reducing the stabilization problem 
to a
Nevanlinna-Pick interpolation problem.
This reduction relies on results stated in 
Lemmas \ref{thm:SStoInt} and \ref{thm:ReductionNPInt}
have appeared in \cite{Ghosh1988, Olbrot1994}, 
but we give new proofs of these results 
based on Theorem \ref{thm:SS_CM}.

First we show that
the stabilization problem is equivalent
to an interpolation problem with a specified codomain:
\begin{lemma}
	\label{thm:SStoInt}
	{\it 
		Let $\underline \theta < 0 < \overline \theta$ and 
		$S := [\underline \theta, \overline \theta] \subset S_{\max}$.
		There exists a controller that stabilizes
		$P_{\theta}$ in \eqref{eq:scalar_TF} 
		for all $\theta \in S$, that is,
		there exists $Q \in \mathcal{RH}^{\infty}$ satisfying
		\eqref{eq:SS_Unimodular} for all $\theta \in S$
		if and only if
		there exists $f \in \mathcal{RH}^{\infty}$ such that 
		$f$ is a map from $\bar{\mathbb{D}}$ to
		$\mathbb{G} := \mathbb{C} \setminus
		\{(-\infty, 1/\underline \theta] \cup [1/\overline \theta,\infty) \}$ and
		satisfies
		the interpolation conditions $f(0) = 0$, $f(1) = 0$, 
		and $f(1/\lambda) = -1$. 
	}
\end{lemma}
\begin{proof}
We obtain an $\mathcal{RH}^{\infty}$ coprime factorization 
$P_{\theta} = N_{\theta}/D_{\theta}$ with
\begin{equation}
\label{eq:NDdef}
N_{\theta}(z) := \frac{(\lambda-1)z}{z-c} - \theta \frac{\lambda z(z-1)}{z-c},
\quad
D_{\theta}(z) := \frac{1-\lambda z}{z-c},
\end{equation}
where $c$ is a fixed complex number with $|c| > 1$. 
If we define $X_0$ and $Y_0$ as in \eqref{eq:XY_def},
then the Bezout identity $N_0X_0 + D_0 Y_0 = 1$ holds.
Hence defining $T_1$ and $T_2$ by \eqref{eq:XY_def},
we see that
$U_{\theta} + V_{\theta}Q$ in \eqref{eq:SS_Unimodular} satisfies
\begin{equation}
\label{eq:separation}
U_{\theta} + V_{\theta}Q
=
1 - \theta 
\left(
T_1
+
T_2Q
\right)
\quad (Q \in \mathcal{RH}^{\infty}).
\end{equation}

Theorem \ref{thm:SS_CM} and \eqref{eq:separation} show that
the plant $P_{\theta}$
is simultaneously stabilizable by a single LTI controller
if and only if there exists $Q \in \mathcal{RH}^{\infty}$ such that
\begin{equation}
\label{eq:thetaT1T2_inv}
(1 - \theta (T_1 + T_2Q))^{-1} \in \mathcal{RH}^{\infty}\qquad
(\theta \in S).
\end{equation}
We have \eqref{eq:thetaT1T2_inv} if and only if 
$1 - \theta (T_1 + T_2Q)$ has no zero in $\bar{\mathbb{D}}$ 
for all $\theta \in S$, that is,
$
T_1(z) + T_2(z)Q(z) \in \mathbb{G}
$
for all $z \in \bar{\mathbb{D}}$.

It is now enough to show that
$G \in \mathcal{RH}^{\infty}$ if and only if
$f := T_1 +T_2Q$ satisfies $f \in \mathcal{RH}^{\infty}$ 
and the interpolation conditions in the lemma.

Suppose that $Q\in \mathcal{RH}^{\infty}$. Since
$T_1, T_2 \in \mathcal{RH}^{\infty}$, we have
$
f = T_1 + T_2Q \in \mathcal{RH}^{\infty}.
$
Moreover,
since the unstable zeros of $T_2$ are $0$, $1$, and $1/\lambda$
and since $Q$ has no unstable poles,
it follows that $f(0) = T_1(0) = 0$, $f(1) = T_1(1) = 0$, and
$f(1/\lambda) = T_1(1/\lambda) = -1$.

Conversely, let $f\in \mathcal{RH}^{\infty}$ satisfy 
$f(0) = 0$, $f(1) = 0$, and $f(1/\lambda) = -1$.
If we define 
\begin{equation}
\label{eq:U_0andF0}
Q := \frac{f- T_1}{T_2},
\end{equation}
then $Q$ belongs to $\mathcal{RH}^{\infty}$.
Assume, to get a contradiction, 
that $Q \not\in \mathcal{RH}^{\infty}$. Since
\begin{equation}
\label{eq:QT2}
T_2Q = f - T_1\in \mathcal{RH}^{\infty},
\end{equation} 
it follows that
$Q$ has some
unstable poles that are zeros of $T_2$ in $\bar{\mathbb{D}}$. 
Let $p_0$ be one of the poles.
Since $T_2$ has only simple zeros in $\bar{\mathbb{D}}$, it follows that
$(T_2Q)(p_0) \not= 0$. The interpolation conditions of $f$ lead to
$f(p_0) - T_1(p_0) = 0$, which contradicts the equality in \eqref{eq:QT2}.
\end{proof}

From Lemma \ref{thm:SStoInt},
it suffices to study the following interpolation problem
for stabilizability:
\begin{problem}
	\label{prob:interpolation}
	{\it
		Let $z_1,\dots, z_n$ be distinct points in $\bar{\mathbb{D}}$ and
		let $w_1,\dots,w_n$ belong to
		$\mathbb{G} := \mathbb{C} \setminus
		\{(-\infty, 1/\underline \theta] \cup [1/\overline \theta,\infty) \}$.
		Find a function $f \in \mathcal{RH}^{\infty}$ such that
		\begin{align} \label{eq:F_int_cond}
		f:~\bar{\mathbb{D}} \to \mathbb{G} \text{~~~and~~~}
		f(z_i) = w_i \quad (i=1,\dots,n).
		\end{align} 
	}
\end{problem}

We solve Problem \ref{prob:interpolation}
by reducing it to the Nevanlinna-Pick interpolation.
To this effect,
we need a conformal map from 
$
\mathbb{G}
$
to $\mathbb{D}$.
In \cite{Olbrot1994}, 
\cite[Section 4.1]{foias1996},
such a conformal map $\phi$ is given in
\eqref{eq:conformal_map}.

Using
the conformal map defined in \eqref{eq:conformal_map},
we see that
Problem \ref{prob:interpolation} can be reduced to
the Nevanlinna-Pick interpolation problem.
\begin{lemma}
	\label{thm:ReductionNPInt}
	{\it
		Problem \ref{prob:interpolation} is solvable if and only if
		the following Nevanlinna-Pick interpolatoin problem is solvable:
		Find a function $g \in \mathcal{RH}^{\infty}$ 
		such that
		\begin{align} \label{eq:G_int_cond}
		g:~\bar{\mathbb{D}} \to \mathbb{D} \text{~~~and~~~}
		g(z_i) = \phi(w_i) \quad (i=1,\dots,n).
		\end{align}
	}
\end{lemma}

\begin{proof}
Let $f$ be a solution to Problem \ref{prob:interpolation}, and
set $g=\phi \circ f$. Then we derive
the equivalence between \eqref{eq:F_int_cond} and \eqref{eq:G_int_cond}.
Since $\phi$ is a conformal map,
we see that $f$ is holomorphic in $\bar {\mathbb{D}}$ if and only if
$g$ is so.

Regarding the rationality of solutions, 
since $\phi^{-1}$ is given by
a rational function in \eqref{eq:phi_inv}
and since $f = \phi^{-1} \circ g$, 
it follows that $f$ is rational for every rational solution $g$.

Conversely, $\phi \circ f$ may not be rational for a
rational function $f$. 
However, $\phi \circ f$ is holomorpic in $\bar {\mathbb{D}}$, 
and hence it is an irrational
solution of
the Nevanlinna-Pick interpolation problem.
If the Nevanlinna-Pick interpolation problem is solvable, 
then there exists a rational solution,
which can be obtained from the Schur-Nevanlinna algorithm, e.g., in
\cite{luxemburg2010, wakaiki2012}
and the explicit formula of the solutions in 
\cite[Sec. 2.11]{foias1996}.
We therefore have the desired rational function $g$.
\end{proof}

Interpolating functions and a conformal map in Lemma \ref{thm:ReductionNPInt} 
are illustrated by the commutative diagram in Fig.~\ref{fig:diagram}.

\begin{figure}[tb]
	\centering
	\includegraphics[width = 6cm,clip]{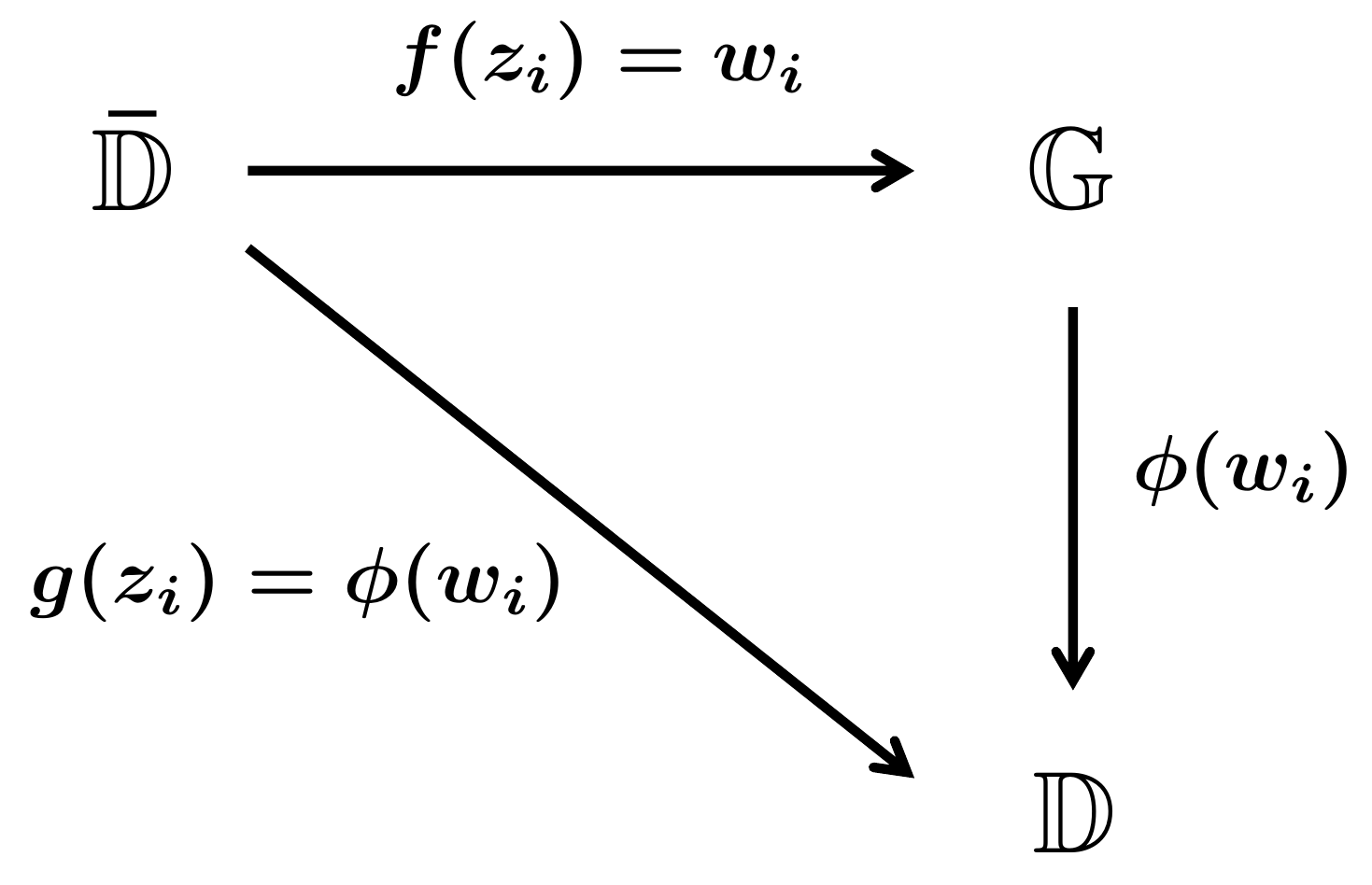}
	\caption{ Interpolating functions $f,g$ and a conformal map $\phi$.}
	\label{fig:diagram}
\end{figure}

Finally we obtain the proof of Theorem \ref{thm:scalar_LTIcontroller}.

\noindent\hspace{1em}{\itshape 
	{\bf Proof of Theorem \ref{thm:scalar_LTIcontroller}:} }
Lemmas \ref{thm:SStoInt} and \ref{thm:ReductionNPInt} show that
the stabilization problem for systems with clock offsets
can be reduced to 
the Nevanlinna-Pick interpolation problem 
with a boundary condition; see, e.g., 
\cite[Sec. 2.11]{foias1996} for the interpolation problem.
We therefore obtain a necessary and sufficient condition 
based on the positive definiteness of
the associated Pick matrix:
\begin{equation}
\label{eq:PickMatrix}
\begin{bmatrix}
1 & 1 \\
1 & \frac{1 -  | \phi(-1)|^2}{1 - 1/\lambda^2}
\end{bmatrix} > 0.
\end{equation}
From the Schur complement formula, \eqref{eq:PickMatrix} is equivalent to
\begin{equation}
\label{eq:kappa_p_1}
\frac{1 - |\phi(-1)|^2}{1 -1/\lambda^2} > 1.
\end{equation}
We see that \eqref{eq:kappa_p_1} is 
\begin{equation}
\label{eq:p_sqrt}
\lambda^2 
\left|
1 - \sqrt{ \frac{1+\overline \theta }{1+\underline \theta } }
\right|^2 < 
\left|
1 + \sqrt{ \frac{1+\overline \theta }{1+\underline \theta }}
\right|^2.
\end{equation}
Since $-1 < \underline \theta < \overline \theta$, 
it follows that \eqref{eq:p_sqrt} is equivalent to
\begin{equation}
\label{eq:1_theta}
(p-1)^2 (1+\overline \theta ) < (p+1)^2 (1+\underline \theta ).
\end{equation}
After rearranging this, we derive \eqref{eq:M1M2cond}.
\hspace*{\fill} $\blacksquare$

\noindent\hspace{1em}{\itshape {\bf Proof of Corollary 
		\ref{coro:maximum_length}:} }
Substituting $\underline \theta = e^{-a\overline{\Delta}} - 1$ and
$\overline \theta  = e^{-a\underline{\Delta}} - 1$ into \eqref{eq:1_theta}, 
we obtain
\begin{equation*}
(\lambda-1)^2 e^{-a\underline{\Delta}} < (\lambda+1)^2 e^{-a\overline{\Delta}},
\end{equation*}
and hence
\begin{equation*}
e^{a(\overline{\Delta} - \underline{\Delta})} < 
\left( \frac{\lambda+1}{\lambda-1} \right) ^2.
\end{equation*}
Taking the logarithm function of both sides gives the desired conclusion.
\hspace*{\fill} $\blacksquare$

\subsection{Comparison with time-invariant/2-periodic 
	static controllers}
The proposition below gives the exact bound
on the clock offset that could be obtained using a static
stabilizer for a scalar plant.
\begin{proposition}
	\label{prop:static}
	{\it
		Consider the extended system \eqref{eq:scalar_sys}.
		Define $\lambda := e^{ah}$, $\theta := e^{-a\Delta} - 1$, and
		$S_{\max}$ as in \eqref{eq:Imax_def}.
		There exists  a static output feedback controller
		$u_k = -Ky_k$ that achieves $\lim_{k \to \infty}\xi_k = 0$
		for every $\theta \in S \subset S_{\max}$ 
		if and only if
		\begin{equation}
		\label{eq:static_bound}
		S \subset \left(-\frac{1}{\lambda},~\frac{1}{\lambda}\right).
		\end{equation}
	}
\end{proposition}
\begin{proof}
	Without loss of generality, we assume that $b/a = 1$.
	Introducing the static controller
	\[
	u_k=-K
	\begin{bmatrix}
	0 & 1
	\end{bmatrix}
	\xi_k
	\]
	into the extended system \eqref{eq:scalar_sys}, we have
	\begin{align}
	\xi_{k+1}=\left[
	\begin{array}{cc}
	-\lambda \theta & - \lambda\theta(1-K)\\
	\lambda (1+\theta) &  \lambda (1+\theta)(1-K)+K
	\end{array}\right]\xi_k.\label{extsys_static}
	\end{align}
	From the Jury stability criterion, the above system is stable if and only
	if the following three inequalities hold:
	\begin{align}
	(\lambda-1)(K-1) &>0\label{staticJury-1}\\
	\lambda\theta K+1 &>0\label{staticJury-2}\\
	2\lambda\theta K+(\lambda-1)K-\lambda-1 &<0\label{staticJury-3}.
	\end{align}
	
	From (\ref{staticJury-1}) and $\lambda>1$, we have
	\begin{align}
	K>1.\label{staticJury-1-1}
	\end{align}
	Therefore \eqref{staticJury-2}
	and \eqref{staticJury-3} give a lower and upper bound on $\theta$,
	respectively:
	\begin{align}
	\label{staticJury-2-1}
	\theta>-\frac{1}{\lambda K},\quad
	\theta<\frac{p+1}{2\lambda K}-\frac{\lambda-1}{2\lambda}.
	\end{align}
	Notice that the lower (upper)
	bound in (\ref{staticJury-2-1}) is increasing (decreasing) with
	respect to $K$.
	Hence these bounds take the inifimum and the supremum under
	(\ref{staticJury-1-1}) when $K \to 1$, and
	\begin{align}
	\lim_{K \to 1}\left(-\frac{1}{\lambda K}\right)
	=-\frac{1}{\lambda},~~
	\lim_{K \to 1}\left(\frac{p+1}{2\lambda K}-
	\frac{\lambda-1}{2\lambda}\right)
	=\frac{1}{\lambda}.\notag
	\end{align}
	Thus,
	there exists a
	static controller $K$ that stabilizes 
	the extended system \eqref{eq:scalar_sys}
	for all
	$\theta\in S$ if and only if $S$ satisfies \eqref{eq:static_bound}.
\end{proof}

The next result provides a sufficient condition on the offset range
$(\underline \theta, \overline \theta)$
for the existence of time-varying
2-periodic controllers that 
stabilize the extended system \eqref{eq:scalar_sys}.
\begin{proposition}
	\label{prop:static_periodic}
	{\it
		Consider the extended system \eqref{eq:scalar_sys}.
		Define $\lambda := e^{ah}$, $\theta := e^{-A\Delta} - 1$,
		and $S_{\max}$ as in \eqref{eq:Imax_def}.
		There exists a 2-periodic static controller
		\begin{equation}
		\label{eq:class_periodic}
		\begin{bmatrix}
		u_k \\
		u_{k+1}
		\end{bmatrix}
		=
		-
		\begin{bmatrix}
		K_1 & 0 \\
		0 & K_2
		\end{bmatrix}
		\begin{bmatrix}
		y_k \\
		y_{k+1}
		\end{bmatrix}
		\end{equation}
		that achieves $\lim_{k \to \infty}\xi_k = 0$
		for every $\theta \in S \subset S_{\max}$
		if
		\begin{equation}
		\label{eq:2peridic_bound}
		S \subset 
		\left(
		-\frac{1}{\lambda\kappa},~
		\frac{1}{\lambda\kappa}
		\right),
		\end{equation}
		where 
		\begin{equation*}
		\label{def:kappa}
		\kappa := 
		\frac{\sqrt{\lambda^2 + 1} \sqrt{\lambda^2 - 4\sqrt{2} +5} - 2
			\left(\sqrt{2} - 1
			\right)\lambda}{\lambda^2 - 1} < 1.
		\end{equation*}
	}
\end{proposition}

\begin{proof}
This is also based on the Jury stability criterion.

Without loss of generality, we assume that $b/a = 1$.
With the 2-periodic controller \eqref{eq:class_periodic}, 
the extended system \eqref{eq:scalar_sys} can be written as
\begin{align}
\xi_{k+2}=
\left(F^2
-\left[
\begin{array}{cc}
F G & G\\
\end{array}\right]
\left[
\begin{array}{cc}
1 & 0\\
-K_2 H G & 1
\end{array}\right]
\left[
\begin{array}{cc}
K_1 & 0\\
0 & K_2
\end{array}\right]
\left[
\begin{array}{c}
H\\
H F
\end{array}\right]
\right)\xi_k,\label{extsys_2p}
\end{align}
where
\begin{align}
F:=\left[
\begin{array}{cc}
-\lambda\theta & -\lambda\theta\\
\lambda(1+\theta) & \lambda(1+\theta)
\end{array}\right],\quad
G:=\left[
\begin{array}{c}
-\lambda\theta \\
\lambda(1+\theta)-1
\end{array}\right],\quad
H:=\left[
\begin{array}{cc}
0 & 1\\
\end{array}\right].\notag
\end{align}

Denote the characteristic polynomial $\rho(\lambda)$ of the matrix in
(\ref{extsys_2p}) by $\rho(\lambda)=\lambda^2+\alpha_1\lambda+\alpha_0$.
The coefficients $\alpha_0, \alpha_1$ are given as
\begin{align}
&\alpha_0:=\zeta p^2\theta^2,\notag\\
&\alpha_1:=
-\zeta \lambda^2 \theta^2
+(2\zeta-\eta)(1-\lambda)\lambda\theta
-(\zeta-\eta+1)\lambda^2+(2\zeta-\eta)\lambda-\zeta,\notag
\end{align}
where $\zeta:=K_1K_2$ and $\eta:=K_1+K_2$.
From the Jury stability test, we have that stability of (\ref{extsys_2p}) is
equivalent to the following three inequalities i)--iii):
i) The first condition is given by
\begin{align}
&1-\alpha_0
=1-\zeta \lambda^2\theta^2>0.
\label{2pJury-1}
\end{align}

ii) Furthermore,
\begin{align}
1+\alpha_1+\alpha_0
&=(2\zeta-\eta)(1-\lambda)\lambda\theta
-(\zeta-\eta-1)\lambda^2+(2\zeta-\eta)\lambda-\zeta+1\notag\\
&=(1-\lambda)\left\{
(2\zeta-\eta)\lambda\theta
+(\zeta-\eta+1)\lambda-\zeta+1
\right\}>0.\notag
\end{align}
This inequality is equivalent to the following inequalities:
\begin{align}
\begin{cases}
\theta <  -\frac{ (\zeta-\eta+1)\lambda-\zeta+1 }{ (2\zeta-\eta)\lambda }
& \text{if }2\zeta-\eta>0\\
\zeta > 1
& \text{if }2\zeta-\eta=0\\
\theta >  -\frac{ (\zeta-\eta+1)\lambda-\zeta+1 }{ (2\zeta-\eta)\lambda }
& \text{if }2\zeta-\eta<0.
\end{cases}
\label{2pJury-2}
\end{align}

iii) Finally,
\begin{align}
1-\alpha_1+\alpha_0
=2\zeta \lambda^2\theta^2
+\lambda(\lambda-1)(2\zeta-\eta)\theta
+(\zeta-\eta+1)\lambda^2-(2\zeta-\eta)\lambda+\zeta+1 
>0.
\label{2pJury-3}
\end{align}

In what follows, we fix a controller, or $(\zeta,\eta)$, and then evaluate the
range of permissible $\theta$ with the controller.
Suppose that $\zeta>0$ and $2\zeta-\eta<0$.
For such parameters $(\zeta,\eta)$, (\ref{2pJury-1}) and (\ref{2pJury-2}) are
reduced to
\begin{align}
&-\frac{1}{\sqrt{\zeta}\lambda} <\theta< \frac{1}{\sqrt{\zeta}\lambda}
\label{2pJury-1-1}
\end{align}
and
\begin{align} 
\theta > -\frac{ (\zeta-\eta+1)\lambda-\zeta+1 }
{ (2\zeta-\eta)\lambda },
\label{2pJury-2-1}
\end{align}
respectively.
Select the parameters $(\zeta,\eta)$ so that the lower bounds on $\theta$ in
(\ref{2pJury-1-1}) and (\ref{2pJury-2-1}) coincide with each other.
That is, $\zeta$ and $\eta$ are chosen to satisfy the following relation:
\begin{align}
\frac{1}{\sqrt{\zeta}}=\frac{ (\zeta-\eta+1)\lambda-\zeta+1 }{2\zeta-\eta}.
\label{LBcond}
\end{align}
Moreover, we select $(\zeta,\eta)$ so that (\ref{2pJury-3}) holds for any
$\theta \in \mathbb{R}$.
This implies that
\begin{align}
(\lambda-1)^2(2\zeta-\eta)^2
-8\zeta \left\{(\zeta-\eta+1)p^2-(2\zeta-\eta)\lambda+\zeta+1\right\}<0.
\label{2pJury3-1}
\end{align}
With the above class of controllers,
where $\zeta$ and $\eta$ satisfy $\zeta>0$, $2\zeta-\eta<0$, (\ref{LBcond}),
and (\ref{2pJury3-1}), the conditions (\ref{2pJury-1})--(\ref{2pJury-3}) for
stability hold if and only if
(\ref{2pJury-1-1}) follows.
We will show that $\zeta$ and $\eta$ satisfy $\zeta>0$, $2\zeta-\eta<0$, (\ref{LBcond}),
and (\ref{2pJury3-1}) if and only if
\begin{align}
&\zeta\in(\kappa^2,1),\label{class_gamma}\\
&\eta=
\frac{\sqrt{\zeta}\left\{(\lambda-1)\zeta-2\sqrt{\zeta}+\lambda+1\right\}}
{\sqrt{\zeta}\lambda-1},\label{class_eta}
\end{align}
where $\kappa$ is defined in \eqref{def:kappa}.
We then analyze the bounds on $\theta$ followed by (\ref{2pJury-1-1}) when the
controller belongs the class characterized by (\ref{class_gamma}) and
(\ref{class_eta}).

The equation (\ref{class_eta}) is obtained from (\ref{LBcond}) with
the fact that $\sqrt{\zeta}\lambda-1\neq0$.

We now aim to show (\ref{class_gamma}).
Substituting (\ref{LBcond}) into (\ref{2pJury3-1}) and using $2\zeta-\eta<0$,
we have that (\ref{2pJury3-1}) is satisfied if and only if
\begin{align}
\eta<2\left(\zeta+2\sqrt{\zeta}\frac{\sqrt{2}-1}{\lambda-1}\right).
\label{2pJury3-2}
\end{align}

From (\ref{class_eta}) and (\ref{2pJury3-2}), $\zeta$ satisfies
\begin{align}
\frac{(\sqrt{\zeta}-\kappa)(\sqrt{\zeta}-\kappa')}
{\sqrt{\zeta}\lambda-1}>0,
\label{2pJury-3-3}
\end{align}
where $\kappa'$ is given as
\begin{align}
\kappa':=\frac{-\sqrt{\lambda^2+1}\sqrt{\lambda^2-4\sqrt{2}+5}-
	2(\sqrt{2}-1)\lambda}{\lambda^2-1}.
\notag
\end{align}
Note that $\sqrt{\zeta}-\kappa'>0$ from $\lambda > 1$.
Thus, from (\ref{2pJury-3-3}), $\zeta$ satisfies one of the following two cases:
i) $\sqrt{\zeta}>1/\lambda$ and $\sqrt{\zeta}>\kappa$, or
ii) $\sqrt{\zeta}<1/\lambda$ and $\sqrt{\zeta}<\kappa$.
A routine calculation shows that $\kappa>1/\lambda$.
Thus, i) is reduced to $\sqrt{\zeta}>\kappa$ and ii) is 
to $\sqrt{\zeta}<1/\lambda$.
On the other hand, from $2\zeta-\eta<0$ and (\ref{class_eta}), it follows that
$\sqrt{\zeta}>1/\lambda$ and $\zeta<1$.
Therefore, we arrive at (\ref{class_gamma}).

Conversely, (\ref{class_gamma}) implies that $\zeta>0$.
Moreover, from (\ref{class_eta}), we have (\ref{LBcond}) and
\begin{align}
2\zeta-\eta
=\frac{\sqrt{\zeta}(\zeta-1)(\lambda+1)}
{\sqrt{\zeta}\lambda-1}.\notag
\end{align}
The right-hand side is negative by (\ref{class_gamma}) and the fact
$\kappa>1/\lambda$.
Note also that (\ref{2pJury3-1}) holds if $\zeta$ and $\eta$ are taken as
(\ref{class_gamma}) and (\ref{class_eta}).

Finally, taking the supremum and the infimum of the upper bound and the lower
bound on $\theta$ in (\ref{2pJury-1-1}) over (\ref{class_gamma}), we conclude
the proof.
\end{proof}

We are now in a position to compare the bounds
\eqref{eq:LTI_bound_M}, \eqref{eq:static_bound}, and
\eqref{eq:2peridic_bound}.
For all $\lambda = e^{ah} > 1$, a routine calculation shows that 
\begin{equation}
\label{eq:comparison_static_periodic}
\frac{1}{\lambda} < \frac{1}{\lambda\kappa} < \frac{2\lambda}{\lambda^2+1}.
\end{equation}
As expected, the offset condition \eqref{eq:static_bound}
for time-invariant static controllers 
results in the smallest range for values of $\theta = 
e^{-a\Delta} - 1$
because the set of all time-invariant static controllers
is a subset of the class of LTI controllers and 
that of 2-periodic static controllers in \eqref{eq:class_periodic}.
On the other hand,
2-periodic static stabilizers do not belong to the class of LTI controllers,
and vice versa.
The second inequality in \eqref{eq:comparison_static_periodic}
always holds for all $\lambda  > 1$, but the bound \eqref{eq:comparison_static_periodic}
is a sufficient condition.
In order to compare 
the ability to robustly stabilize the closed loop of
2-periodic static controllers versus
LTI controllers, we need to do a brute-force computation for 
the exact bound on clock offset that would be allowed by
a 2-periodic static controller. 

\begin{example}
		Consider a scalar plant with ZOH-update period $h = 1$. 
		Fig.~\ref{fig:comparison} shows the maximum offset lengths 
		$\overline \Delta - \underline \Delta$ allowed by
		LTI stabilizers and static ones, which are
		obtained by 
		Theorem \ref{thm:scalar_LTIcontroller}
		and
		Proposition \ref{prop:static}, respectively. 
		The figure also 
		gives a lower bound on the maximum offset length obtained
		using 2-periodic static stabilizers,
		which is derived from Proposition \ref{prop:static_periodic}.
		All lines decreases exponentially as the unstable pole 
		$a$ grows to $\infty$.
		If the unstable pole $a$ is smaller than 0.9, then
		Assumption \ref{ass:sample} gives the bound 
		$\overline{\Delta} - \underline{\Delta} < 2h = 2$.
		We also observe that LTI controllers double the robustness 
		with respect to that
		achieved by time-invariant static controllers.
		For example, for the unstable pole $a = 1$,
		the maximum offset length by LTI controllers
		is $\overline \Delta - \underline \Delta = 1.544$,
		whereas that by time-invariant static controllers is 
		$\overline \Delta - \underline \Delta = 0.7719$. 
		
		\begin{figure}[bt]
			\centering
			\includegraphics[width = 9.5cm]{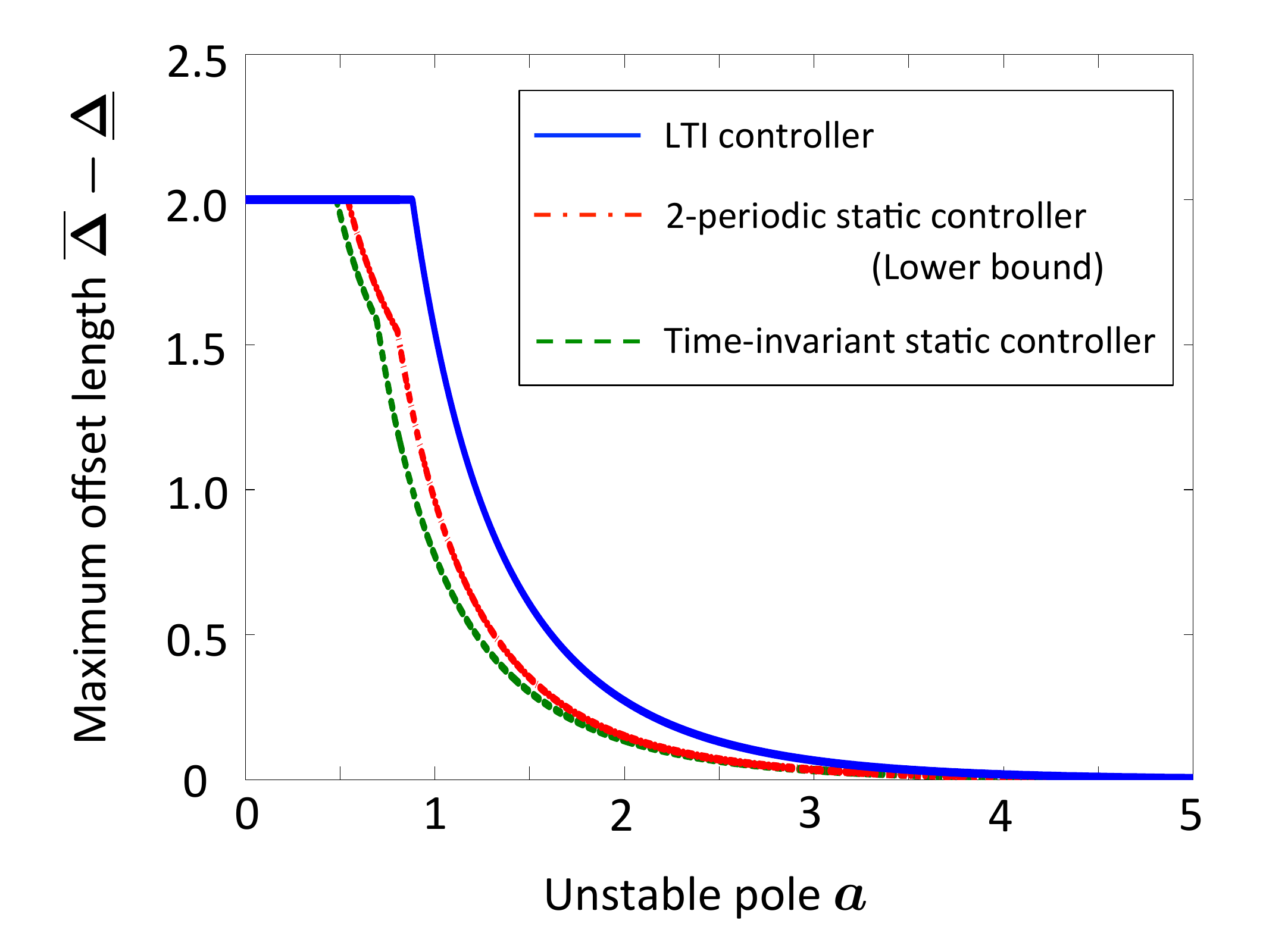}
			\caption{Comparison among the maximum offset length 
				$\overline \Delta - \underline \Delta$ obtained using 
				each class of controllers ($h=1$): 
				The solid (dashed) line is the exact
				bound on the maximum offset length allowed by LTI 
				(time-invariant static) controllers.
				The dashed-dotted line is a lower bound on the maximum offset length
				by 2-periodic static controllers.}
			\label{fig:comparison}
		\end{figure}
\end{example}

\subsection{Conservativeness due to the small gain theorem}
The sufficient condition in Theorem \ref{thm:suff} with
\[
L(z) := \frac{\lambda z(z-1)}{z-c},\qquad R(\theta) := \theta
\qquad
(c\in \mathbb{C},~|c| > 1)
\] in
\eqref{eq:N_cond}, shows that 
there exists a controller stabilizing $P_{\theta}$
for all $\theta \in (-1/\lambda, 1/\lambda)$.
This bound of $\theta$ 
is the same as that was obtained in
\eqref{eq:static_bound} for a static controller,
which shows that 
the use of the small gain theorem is as conservative as
the restriction of controllers to static gains.

This conservativeness arises from the codomain of
interpolating functions.
For simplicity, let the offset bound 
$[\underline \theta, \overline \theta] = [-1,1]$.
In Theorem \ref{thm:suff},
the stabilization problem we consider is equivalent to
finding $f:\bar{\mathbb{D}} \to \mathbb{D}$ satisfying $f(0)=0$,
$f(1)=0$, $f(1/\lambda)=-1$.
Recall that a necessary and sufficient condition for stabilization in 
Lemma \ref{thm:SStoInt}
is the existence of a function 
$f:\bar{\mathbb{D}} \to \mathbb{G}
=
\mathbb{C} \setminus
\{(-\infty, 1] \cup [1,\infty) \}
$ satisfying the same interpolation
conditions.
The difference between the codomains $\mathbb{D}$ and
$\mathbb{G}$ leads to
conservativeness in the stabilization analysis.

\subsection{Regarding a clock offset as an 
	additive uncertainty}
The transfer function $P_{\theta}$ 
in \eqref{eq:scalar_TF} can be viewed
as a perturbation of the nominal transfer function
$P_0 = (\lambda-1)z/(1-\lambda z)$
by the following additive uncertainty:
\[
P_{\theta}(z) - P_0(z) = m(z)r_{\theta }(z)
\]
where
\[
m(z) :=  \frac{z(\lambda-z)}{1-\lambda z},\quad
r_{\theta }(z) := 
\frac{\theta \lambda(z-1)}{\lambda-z}.
\]
Since 
$|m(e^{j\omega})| = 1$,
it follows that
\[
|P_{\theta}(e^{j\omega}) - P_0(e^{j\omega})| \leq 
|W_{\overline \theta }(e^{j\omega})|
\]
for all $\omega \in [0,2\pi]$ and $\theta \in [-\overline \theta, \overline \theta]$
with $\overline \theta > 0$.
Hence, 
as shown in \cite[Sec. 3.5]{foias1996},
there exists a controller stabilizing $P_{\theta}$
for all $\theta \in [-\overline \theta, \overline \theta]$ 
if $\overline \theta$ satisfies
\begin{equation}
\label{eq:robut_one_block}
\left\|
W_{\overline \theta}(X_0 + Q D_0)
\right\|_{\infty} < 1
\end{equation}
for some $Q \in \mathcal{RH}^{\infty}$,
where the $\mathcal{RH}^{\infty}$ functions
$X_0$ and $D_0$ are part of the coprime factorization 
$P_0 = N_0/D_0$ 
satisfying the Bezout identity $N_0X_0+D_0Y_0 = 1$.
Reducing the problem of finding $Q \in \mathcal{RH}^{\infty}$ with
\eqref{eq:robut_one_block} to the Nevanlina-Pick interpolation problem
as in \cite[Sec. 4.3]{foias1996}, we see that 
$\overline \theta$ satisfies
\eqref{eq:robut_one_block} for some $Q \in \mathcal{RH}^{\infty}$
if and only if
$
\overline \theta <  1/\lambda.
$

This bound derived from the classical approach of $\mathcal{H}^{\infty}$ 
robust control turns out to also coincide with the one
in \eqref{eq:static_bound} obtained for a static controller.
This shows that
the classical $\mathcal{H}^{\infty}$ robust control approach uses
an unnecessarily large class of parametric uncertainty. 
In conjunction with the observation in Section 5.2, 
this discussion implies that
the use of the small gain theorem, the over-approximation of a
offset uncertainty
by an $\mathcal{H}^{\infty}$-additive uncertainty, and
the restriction of controllers to static gains have
the same level of conservativeness for scalar plants.

\section{Concluding Remarks}
We studied the problem of stabilizing systems
in which the sensor and the controller have a constant
clock offset. 
We formulated the problem as the stabilization problem for systems with
parametric uncertainty.
For multi-input systems, we derived 
a sufficient condition that 
is numerically testable,
based on the results of simultaneous stabilization.
For first-order systems,
we obtained
the maximum offset length
that can be allowed by an LTI controller.
However, a full investigation of the problem 
for general-order systems 
and systems with model uncertainty  is still 
an open area for future research.

\end{document}